\documentclass[nohyperref]{article}

\usepackage{microtype}
\usepackage{graphicx}
\usepackage{subfigure}
\usepackage{booktabs} 
\usepackage[utf8]{inputenc} 
\usepackage[T1]{fontenc}    
\usepackage{url}            
\usepackage{amsfonts}       
\usepackage{nicefrac}       
\usepackage{xcolor}         
\usepackage[normalem]{ulem} 
\usepackage{amsmath, amssymb, amsthm}
\usepackage{enumerate}
\usepackage{bbm, bm}
\usepackage{times, enumitem}
\usepackage{mathtools}
\usepackage{upgreek}
\usepackage{caption}
\usepackage{ulem}
\usepackage{multicol}
\usepackage{makecell}

\usepackage{hyperref}

\usepackage[accepted]{icml2023}

\usepackage[capitalize,noabbrev]{cleveref}

\theoremstyle{plain}
\newtheorem{theorem}{Theorem}[section]

\newtheorem{lemma}[theorem]{Lemma}

\theoremstyle{definition}
\newtheorem{definition}[theorem]{Definition}

\theoremstyle{remark}
\newtheorem{remark}[theorem]{Remark}

\usepackage[textsize=tiny]{todonotes}

\makeatletter
\newcommand{\customlabel}[2]{%
   \protected@write \@auxout {}{\string \newlabel {#1}{{#2}{\thepage}{#2}{#1}{}} }%
   \hypertarget{#1}{#2}
}
\makeatother

\newcommand{\E}{\mathbb{E}}
\newcommand{\Pb}{\mathbb{P}}

\newcommand{\var}{\text{var}}

\newcommand{\Pn}{\mathbb{P}_n}

\newcommand{\R}{\mathbb{R}}
\DeclareMathOperator*{\argmin}{arg\,min}

\newcommand{\cmmnt}[1]{\ignorespaces}  

\icmltitlerunning{Toward Fair and Robust Policy Learning}

\begin{document}

\twocolumn[
\icmltitle{Fair and Robust Estimation of Heterogeneous Treatment Effects \\ for Policy Learning}




\begin{icmlauthorlist}
\icmlauthor{Kwangho Kim}{ku,hms}
\icmlauthor{Jos\'{e} R. Zubizarreta}{hms,harvard}
\end{icmlauthorlist}

\icmlaffiliation{hms}{Department of Health Care Policy, Harvard Medical School, MA, USA}
\icmlaffiliation{ku}{Department of Statistics, Korea University, Seoul, South Korea}
\icmlaffiliation{harvard}{Departments of Biostatistics and Statistics, Harvard University, MA, USA}

\icmlcorrespondingauthor{Jos\'{e} R. Zubizarreta}{zubizarreta@hcp.med.harvard.edu}

\icmlkeywords{causal inference, conditional average treatment effect, observational study}

\vskip 0.3in
]



\printAffiliationsAndNotice{The authors would like to thank Youmi Suk for many helpful discussions. The proposed algorithm can be applied in \texttt{R} with the code provided in \url{https://github.com/kwangho-joshua-kim/fair-robust-HTE}.}  

\begin{abstract}
We propose a simple and general framework for nonparametric estimation of heterogeneous treatment effects under fairness constraints.
Under standard regularity conditions, we show that the resulting estimators possess the double robustness property. 
We use this framework to characterize the trade-off between fairness and the maximum welfare achievable by the optimal policy. 
We evaluate the methods in a simulation study and illustrate them in a real-world case study.
\end{abstract}

\section{Introduction}
\label{sec:intro}
In today's data-centric world, an increasing number of decisions that affect people's lives are automatically made by machine learning models. 
Such decision-making systems are implemented in various domains, ranging from finance to healthcare. 
Considering the multifaceted implications of such decisions both at the individual and societal levels, it is crucial to ensure that the underlying models are not only accurate but {fair}. 
In this work, by \emph{fairness} we mean that the model estimates are not biased so that they do not systematically benefit or harm a specific group of people. 
The need to address such algorithmic biases has given rise to a plethora of works on algorithmic fairness (e.g., see  \citet[][]{barocas-hardt-narayanan} for a review).
However, despite the abundance of studies in this general area, comparatively little attention has been given to fairness in causal inference. 
In this work, we propose a novel framework for estimating heterogeneous treatment effects in a fair and robust manner, leveraging recent developments in counterfactual optimization \citep{kim2022cfclassifier, kim2022counterfactualmv}.

\subsection{Related Work}

Understanding of treatment-effect heterogeneity and identifying subgroups that respond in a similar way to a given treatment is of great importance across scientific domains. The most common target estimand to study heterogeneity of treatment effects is the conditional average treatment effect (CATE). Various methods have been proposed to obtain accurate estimates of and valid inferences for the CATE, with a special emphasis in recent years on incorporating flexible machine learning tools \citep[e.g.,][]{imai2013estimating, van2014targeted, athey2016recursive, wager2018estimation, kunzel2017meta, nie2017quasi, kennedy2020optimal}. See \citet{jacob2021cate} for a review.

Although CATE estimation and inference are of independent importance, these tasks are central also to better target interventions, for example by learning subgroup structures or optimal policies. Much of the earlier attempts on these areas involve postulating a parametric model for the CATE to find the subgroups that would benefit from the given treatment \citep[e.g.,][]{murphy2003optimal,  brinkley2010generalized, henderson2010regret, orellana2010dynamic}. Optimal policies or treatment rules are often designed to identify subgroups with CATEs larger than a certain threshold of interest, so that investigators can select the most promising subgroups with certain efficacy or safety \citep[e.g.,][]{zhao2013effectively, schnell2016bayesian, chen2017general, ballarini2018subgroup, wang2022causal}.
More robust approaches to find optimal policies based on the property of doubly robustness have been proposed in subsequent research \citep[e.g.,][]{zhang2012robust, zhang2013robust}. 
In recent studies \citep{kallus2018balanced, kitagawa2018should, athey2021policy}, flexible nonparametric approaches are discussed where an optimal policy is deployed from a pre-specified class that encodes problem-specific constraints (e.g., a budget or a capacity constraint). 

However, this type of data-driven policy-making processes may result in discriminatory treatment of groups defined by sensitive features, such as gender or race. In order to mitigate these algorithmic biases, a wide array of fair estimation criteria have been developed by placing restrictions on the joint distribution of outcome variables and sensitive features \citep[e.g.,][]{hardt2016equality, corbett2017algorithmic, barocas-hardt-narayanan}. 
In some cases, such as risk assessment, counterfactual fairness may be of interest where the fairness criteria depend on potential (or counterfactual) outcomes \citep[e.g.,][]{kusner2017counterfactual, nabi2018fair, coston2020counterfactual, mishler2021fairness, mishler2021fade}.
Recently, constraint-based frameworks have been proposed to flexibly incorporate such fairness criteria in classification \citep[e.g.,][]{zafar2019fairness, mishler2021fade}.
It is also known that fairness-accuracy trade-offs may exist because in some cases the most accurate models under consideration do not satisfy the desired fairness criterion \citep[e.g.,][]{kleinberg2016inherent, menon2018cost, obermeyer2019dissecting, mishler2021fairness}. 

To our knowledge, little work has been done at the intersection of algorithmic fairness and the estimation of heterogeneous treatment effects. Some significant studies have integrated aspects of algorithmic fairness and policy learning \citep{nabi2019learning, viviano2022fair}. However, it is still unclear how to extend existing methods in algorithmic fairness to obtain an efficient and robust estimator for the CATE under general fairness criteria. Further, to the best of our knowledge, the trade-off between fairness and the maximum welfare achievable by the optimal policy has not been formally explored.

\vspace*{-.01in}
\subsection{Contribution}
Our method builds on promising literature at the intersection of algorithmic fairness, causal inference, and stochastic optimization, bridging the gap between algorithmic fairness and the analysis of heterogeneous treatment effects. 
At this intersection, our contribution is twofold. 
First, we propose a simple and general framework for nonparametric estimation of the CATE under general fairness constraints. 
We cast our estimator as a convex optimization problem that can be readily solved with off-the-shelf solvers, and show that the resulting estimators are doubly robust under standard regularity conditions. The proposed estimator can attain fast $\sqrt{n}$ rates with tractable inference even when incorporating flexible machine learning tools. Our proposed approach contributes to the existing works in terms of robustness, flexibility, and ease of implementation. 
Second, we characterize the trade-off between welfare and fairness, by analyzing the regret bounds relative to the optimal policy.
To our knowledge, this is the first to quantify the cost of fairness in policy learning, which helps to understand, for example, how a desired level of fairness requires a social welfare compromise. 

\vspace*{-.06in}
\section{Setup and Framework}
\subsection{Heterogeneous Treatment Effects and Policy Learning}

Consider an i.i.d. sample $(Z_{1}, ... , Z_{n})$ of $n$ tuples $Z=(Y,A,S,X) \sim \Pb$ for some distribution $\Pb$, outcome $Y \in \mathcal{Y}$, binary intervention $A \in \{0,1\}$, sensitive feature $S \in \{0,1\}$, and additional covariates $X \in \mathcal{X}$ for some compact subset $\mathcal{X}$. Here, we assume that larger values of $Y$ denote better outcomes. We let $W=(S,X) \in \mathcal{W}$ represent the measured pre-intervention variables and let $Y^a$ denote the potential outcome that would have been observed (possibly contrary to fact) under treatment or intervention $A=a$. Throughout, we assume the common causal identification assumptions of \emph{consistency}, \emph{no unmeasured confounding}, and \emph{positivity} \citep[e.g.,][Chapter 12]{imbens2015causal}. Under these assumptions, the CATE is defined and identified as
\begin{align*}
    \tau(W) = \E(Y^1 - Y^0 \mid W) 
    = \mu_1(W) - \mu_0(W),
\end{align*}
where $\mu_a(W) = \E[Y \mid W, A=a], \forall a \in \{0,1\}$. A treatment rule or \emph{policy} $g$ is defined by a mapping from the pre-treatment variables to the treatment: i.e., $g:\mathcal{W} \rightarrow \{0,1\}$. The CATE function $\tau$ can be used to produce the optimal policies or to identify subgroups of interest. A policymaker often chooses a policy in such a way that the expected utility or \emph{welfare} defined by
\begin{align} \label{eqn:utility}
    \mathcal{U}(g) = \E\left\{Y^1 g(W) + Y^0\left(1 - g(W) \right) \right\}
\end{align}
is maximized. It is straightforward to show that the {optimal policy} producing the largest value of $\mathcal{U}(g)$ is given by
\begin{align} \label{eqn:optimal-trt-regime}
    g^*(W) = \mathbbm{1}\left\{\tau(W) > 0 \right\}.
\end{align}
$g^*$ targets individuals that would have yielded the larger mean outcome had the treatment assigned\footnote{Here, the strict inequality follows from the convention \citep[see, e.g., ][]{zhang2012robust}.}. More generally, however, one may aim to learn a targeting policy
\begin{align} \label{eqn:subgroup-rule}
    g_{\mathcal{I}}(W) = \mathbbm{1}\left\{\tau(W) \in \mathcal{I} \right\},
\end{align}
where the interval $\mathcal{I}$ defines a subgroup of interest.

\subsection{Motivating Example} \label{sec:motivating-example}
As we illustrate here, ignoring fairness considerations in the CATE estimation problem may introduce serious biases and in turn result in unfair policies. Consider the following data-generating process
\begin{align*}
    & S \sim Bernoulli(0.5), \,\, [X_1, X_2]^\top \mid S \sim \mathcal{N}([0, 2S-1]^\top, I_2), \\
    & \Pb(A=1 \mid W) = expit\left(W^\top [1,0,0] + SX_1 \right), \\
    & \mu_A(W) = AX_2^3/2 + f_\mu(W),
\end{align*}
for some fixed function $f_\mu:\mathcal{W} \rightarrow \R$,
where $expit$ and $I_2$ denote the inverse logit function and the $2 \times 2$ identity matrix. Then, $\tau(W) = X_2^3/2$ and $g^*(W)=\mathbbm{1}(X_2>0)$. When we generate $100$ observations from this model, serious fairness issues arise, as can be seen in Figure \ref{fig:motivating-example}. For example, under $g^*$ only less than $5$\% of individuals with $S=0$ are treated, while more than $90$\% of individuals in the untreated group are $S=0$.

\begin{figure*}[!t]
\centering
\begin{minipage}{.425\linewidth}
  \centering
  \includegraphics[width=\linewidth]{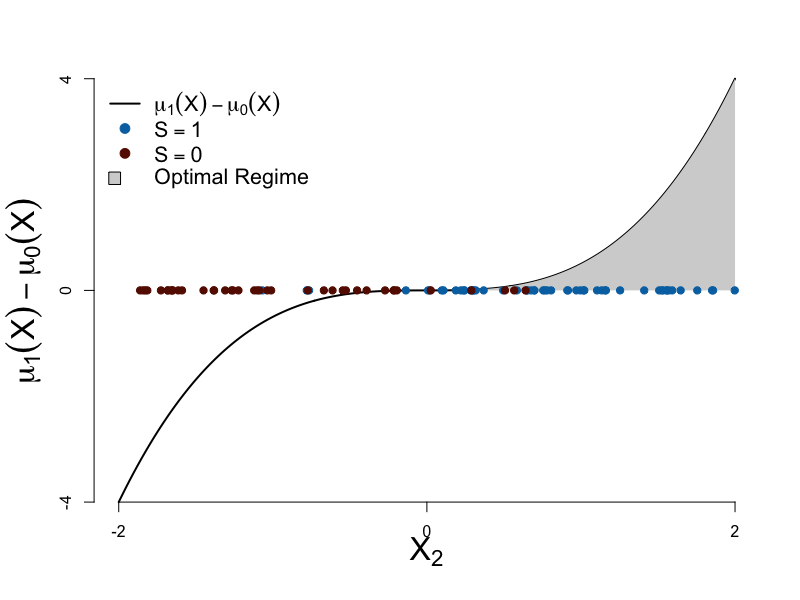}
\end{minipage}
\hfill
\begin{minipage}{.43\linewidth}
  \centering
  \includegraphics[width=\linewidth]{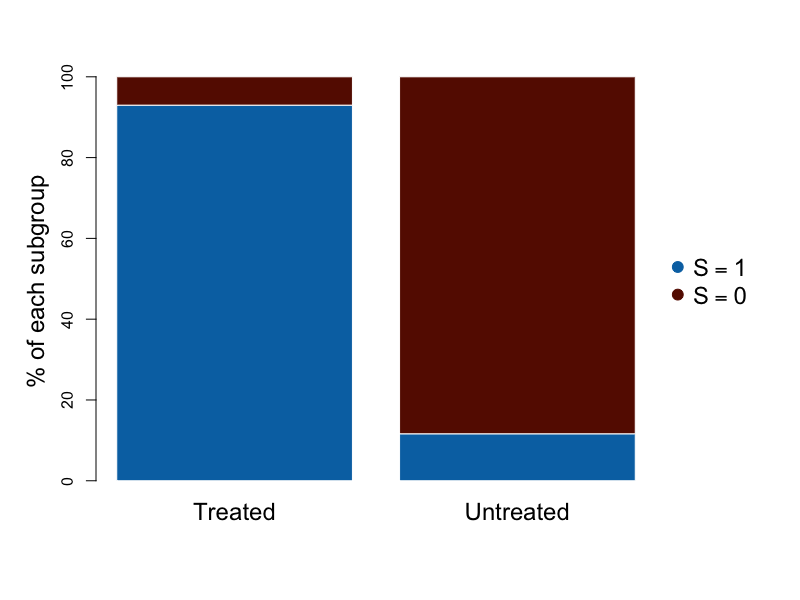}  
\end{minipage}
\caption{When the optimal policy is applied, only less than $7$\% of individuals with $S=0$ are treated while more than $95$\% of individuals in the untreated group are $S=0$.}  \label{fig:motivating-example}
\end{figure*}

In this case, one would be inclined to believe that the sensitive feature $S$ is the source of the heterogeneity. Also, this may lead to policies that discriminate against individuals belonging to the group $S=0$. Hence, accurate estimation of $\tau$ alone may induce socially unacceptable consequences which can cause negative impacts on certain individuals. For example, if $A$ represents the policing practice of stop-and-frisk program, it could be used as a recipe for discriminatory practice of stop-and-frisk toward specific ethnic groups \citep{gelman2007analysis, white2016stop}. If $A$ represents access to medical or health care resources, it could be used as an excuse to create a barrier to healthcare for people that are socially disadvantaged. Many recent studies advocate not only cost-effectiveness but also other ethical values for rationing limited health resources \citep[e.g.,][]{emanuel2020fair, obermeyer2019dissecting}. Unless properly addressed, the biases in estimated heterogeneity treatment effects can perpetuate and amplify unfair societal processes and outcomes. As a result, this can reduce the public’s trust in data-driven systems \citep{schwartz2022towards}.

\subsection{Proposed Framework}

In this section, we lay out a framework for estimating the CATE where we can maximize model fit while securing algorithmic fairness. Specifically, we aim to find a functional approximation of $\tau$, defined by a projection onto a finite-dimensional parametric model subject to fairness constraints. Our target parameter can be reformulated as the solution to the following constrained stochastic optimization problem
\begin{equation}
\label{eqn:true-opt-problem}
\begin{aligned} 
    & \underset{\beta \in \R^k}{\text{minimize}} \quad \mathcal{R}_{\text{MSE}}\left(\beta;\bm{b}\right) \coloneqq \E\left[ \left\{ Y^1 - Y^0 - \beta^\top\bm{b}(W) \right\}^2\right] \\
    & \text{subject to } \quad \left\vert \E\left\{\text{uf}_j(Z)\beta^\top\bm{b}(W)\right\} \right\vert \leq \delta_{j}, \, j \in J,     
\end{aligned}     \tag{$\mathsf{P}$}  
\end{equation}
for some $\delta_{j} \geq 0$ and $J = \{1,...,m\}$. $\delta_{j}$ is a prespecified tolerance for the maximum acceptable level of unfairness in the $j$-th criterion. The solution of the above program corresponds to the coefficients of the best-fitting function of $\tau$ on the finite-dimensional model space spanned by the basis functions $\bm{b}(W)=[b_1(W),...,b_k(W)]^\top$ subject to the $m$ fairness constraints. Note that here we do not assume anything about the true functional relationship between $Y^1-Y^0$ and $W$. This form of linear aggregation is widely used in nonparametric regression \citep[e.g.,][]{juditsky2000functional, tsybakov2003optimal}. \citet{zhao2019covariate} discusses adaptive strategies that allow to consider very rich model family using $\bm{b}(\cdot)$.

We follow \citet{mishler2021fade} and use the \emph{fairness function} $\text{uf}_j: \mathcal{Y}\times\mathcal{W}\times\{0,1\} \rightarrow \R$ to accommodate a broad range of (counterfactual) fairness measures. For example, the criterion of \emph{independence} or \emph{statistical parity}, arguably the most well-known criterion for fairness, can be applied by letting
\begin{align} \label{eqn:g-independence}
    \text{uf}_j(Z) = \frac{1-S}{\E(1-S)} - \frac{S}{\E(S)}.
\end{align}
This leads to \[\left\vert \E\left\{\beta^\top\bm{b}(W) \mid S=0 \right\} - \E\left\{\beta^\top\bm{b}(W) \mid S=1 \right\} \right\vert \leq \delta_{j},\] thus requiring our fitted models to be marginally (approximately) independent of the sensitive feature with respect to the CATE. One could protect the sensitive variable $S$ using \emph{conditional statistical parity} with the fairness function
\begin{align} \label{eqn:g-conditional-independence}
    \text{uf}_j(Z) = \frac{(1-S)\mathbbm{1}\left\{L = l\right\}}{\E\left[(1-S)\mathbbm{1}\left\{L = l\right\}\right]} - \frac{S\mathbbm{1}\left\{L = l\right\}}{\E\left[S\mathbbm{1}\left\{L = l\right\}\right]},
\end{align}
where $L$ is a set of legitimate factors used to specify the conditional statistical parity, and $L$ is some
function of $X$ \citep{corbett2017algorithmic}. This leads to
\begin{align*}
& \big\vert \E\left\{\beta^\top\bm{b}(W) \mid S=0, L = l \right\} \\
& \, - \E\left\{\beta^\top\bm{b}(W) \mid S=1, L = l \right\} \big\vert \leq \delta_{j},
\end{align*}
meaning that controlling for a
limited set of “legitimate” risk factors, our fitted model is conditionally independent of the sensitive feature. When we want to pursue independence for the group that benefits from the treatment, we may consider 
\begin{align} \label{eqn:g-positive-balance}
    \text{uf}_j(Z)
    = \frac{(1-S)\mathbbm{1}\left\{\tau(W)>0\right\} }{\E\left[(1-S)\mathbbm{1}\left\{\tau(W)>0\right\}\right]}
    - \frac{S\mathbbm{1}\left\{\tau(W)>0\right\}}{\E\left[S \mathbbm{1}\left\{\tau(W)>0\right\}\right]}, 
\end{align}
which leads to 
\begin{align*}
& \big\vert \E\left\{\beta^\top\bm{b}(W) \mid S=0, \tau(W)>0 \right\} \\
& \, - \E\left\{\beta^\top\bm{b}(W) \mid S=1, \tau(W)>0 \right\} \big\vert \leq \delta_{j}.
\end{align*}
This corresponds to the condition of \emph{balance for the positive class} \citep{kleinberg2016inherent}, with $\tau(W)$ viewed as a risk score.
One may also employ other fairness criteria with appropriate fairness functions. In case of counterfactual fairness, each $\text{uf}_j$ depends on $Y^1, Y^0$ as well \citep[][Section 3]{mishler2021fade}.

Similar projection approaches to \eqref{eqn:true-opt-problem} have also been used in causal inference \citep[e.g.,][]{neugebauer2007nonparametric, semenova2021debiased, kennedy2021semiparametric}.
There are several reasons why the above projection approach is preferred in our setting. First, as will be seen shortly, the coefficients $\beta$ may be estimated with flexible nonparametric methods while achieving the parametric $\sqrt{n}$ rates and tractable inference due to the property of double robustness. It also aids interpretability: it allows practitioners to better grasp the heterogeneity in treatment effects and audit the resulting policies according to the specified level of unfairness. Finally, the optimal solution of \eqref{eqn:true-opt-problem} can be estimated by solving the optimization problem that approximates \eqref{eqn:true-opt-problem} using various off-the-shelf algorithms. We will discuss this in more detail in the following section.

Another helpful feature of this framework is that we can consider a general setting where only a subset of covariates $V \subseteq W$ can be used for predicting the counterfactual contrast $Y^1-Y^0$. This allows for \textit{runtime confounding}, where some factors used by decision-makers are recorded in the training data (used to construct nuisance estimates) but not available for prediction (see \citet{kim2022cfclassifier} for details; see also \citet{coston2020RuntimeConfounding} and the references therein).

\begin{remark}
    Following the convention in the fairness literature, in \eqref{eqn:true-opt-problem} we only place restrictions on the absolute difference between the first moments (i.e., the conditional means). However, one may consider restricting more comprehensive set of distributional features beyond the means by replacing the constraints with
    \begin{align*}
        \left\vert \E\left\{\text{uf}_j(Z)\text{d}\left(\beta^\top\bm{b}(W)\right) \right\} \right\vert \leq \delta_{j},
    \end{align*}
    where the function $\text{d}:\R \rightarrow \R^p$ determines the distributional features which we want to constrain in each fairness criterion. For instance, one can match up to the $p$-th moment by setting $\text{d}(x)=[x, \ldots, x^p]^\top$. This generalizes the conventional measure of fairness and allows to pursue stricter levels of fairness. However, this entails costs of much harder implementation (i.e., non-convex and non-linear optimization) and stronger regularity conditions for our theoretical results. We leave this extension for future work.
\end{remark}

\textbf{Notation.} 
We briefly introduce some additional notation used in the rest of the paper. For any fixed vector $v$, we let $\Vert v \Vert_q$ denote the $L_q$-norm. Let $\Pn$ denote the empirical measure over $(Z_1,...,Z_n)$. Given a sample operator $h$ (e.g., an estimated function), we let $\Pb$ denote the conditional expectation over a new independent observation $Z$, as in $\Pb(h)=\Pb\{h(Z)\}=\int h(z)d\Pb(z)$ \footnote{When $h$ is a fixed operator, $\Pb$ and $\E$ are used interchangeably.}. Then we use $\Vert h \Vert_{q,\Pb}$ to denote the $L_q(\Pb)$ norm of $h$ defined by $\Vert h \Vert_{q,\Pb} = \left[\int \Vert h(z) \Vert_q^q d\Pb(z)\right]^{\frac{1}{q}}$. Lastly, we let $\lesssim$ denote less than or equal to up to a nonnegative constant.

\section{Estimation and Inference}
\eqref{eqn:true-opt-problem} is not directly solvable so we need to find an \emph{approximating program} of the ``true" program \eqref{eqn:true-opt-problem}. A complication arises since standard approaches to stochastic programming such as \textit{stochastic approximation} (SA) and \textit{sample average approximation} (SAA) \citep[e.g.,][]{nemirovski2009robust, shapiro2014lectures} are infeasible in our setting, because i) the relevant sample moments and stochastic (sub)gradients depend on unobserved counterfactuals, and ii) these approaches cannot incorporate efficient semiparametric estimators with cross-fitting \citep[e.g.,][]{ newey2018cross}. We therefore build our estimators on the recent developments by \citet{kim2022cfclassifier, kim2022counterfactualmv} where counterfactual components are estimated flexibly without any restrictions on our estimand or estimator.

We first describe simple estimators of each counterfactual component. 
Based on the identification assumptions, we obtain the following identity
\begin{align*}
    \E[Y^a\bm{b}(W)] = \E\left[ \mu_a(W) \bm{b}(W) \right] = \E\left[ \frac{AY}{\pi_a(W)} \bm{b}(W) \right],
\end{align*}
where $\pi_a(W)=\Pb[A=a \mid W]$ is the \emph{propensity score}. Then one may estimate the counterfactual parameter $\E[Y^a\bm{b}(W)]$ using the \emph{plug-in} (PI) estimator $\Pn\left\{ \widehat{\mu}_a(W) \bm{b}(W) \right\}$ or the \emph{inverse-probability-weighted} (IPW) estimator $\Pn\left\{ {AY}/{\widehat{\pi}_a(W)} \bm{b}(W) \right\}$ depending on the quality of information to model the observational outcome or treatment process. 
Here $\widehat{\mu}_a$ and $\widehat{\pi}_a$ are some estimators of ${\mu}_a$ and $ {\pi}_a$, respectively. Although widely used in practice, these estimators cannot attain $\sqrt{n}$ rates in general when nonparametric methods are used \citep{kennedy2016semiparametric}. 

We provide more efficient influence-function-based semiparametric estimators for the counterfactual components in \eqref{eqn:true-opt-problem}. Let $\varphi_a$ denote the uncentered efficient influence function (EIF) for the parameter $\E[Y^a]=\E\left\{\E[Y \mid W,A=a ]\right\}$, which is defined by 
\begin{align*}
    & \varphi_a(Z;\eta) = \frac{\mathbbm{1}(A=a)}{\pi_a(W)}\left\{Y- \mu_A(W)\right\} + \mu_a(W),
\end{align*}
with a set of the nuisance components $\eta=\{\pi_a, \mu_a\}$ \citep{kennedy2017semiparametric, kennedy2022semiparametric}.

Following \citet{robins2008higher}, \citet{zheng2010asymptotic}, \citet{Chernozhukov17}, \citet{newey2018cross}, and \citet{kennedy2020optimal}, we propose to use \emph{sample splitting} (or \emph{cross fitting}) to allow for arbitrarily complex nuisance estimators $\widehat{\eta}$. Specifically, we split the data into $K$ disjoint groups, each with size $n/K$ approximately, by drawing variables $(B_1,..., B_n)$ independent of the data, with $B_i=b$ indicating that subject $i$ was split into group $b \in \{1,...,K\}$. Then the semiparametric estimator for $\E[Y^a\bm{b}(W)]$ based on the EIF and sample splitting are given by
\begin{align*} 
    & \frac{1}{K}\sum_{b=1}^K \Pn^b\left\{ \varphi_a(Z;\widehat{\eta}_{-b})\bm{b}(W)\right\} \equiv \Pn\left\{  \varphi_a(Z;\widehat{\eta}_{-B}) \bm{b}(W) \right\},
\end{align*}
where we let $\Pn^b$ denote empirical averages only over the set of units in group $b$ $\{i : B_i=b\}$ and let $\widehat{\eta}_{-b}$ denote the nuisance estimator constructed only using those units $\{i : B_i \neq b\}$. This suggests
\begin{align} \label{eqn:estimator-objective}
    \Pn\left\{  \varphi_1(Z;\widehat{\eta}_{-B}) - \varphi_0(Z;\widehat{\eta}_{-B}) \bm{b}(W) \right\}
\end{align}
as our estimator for $\E[(Y^1-Y^0)\bm{b}(W)] = \E[\tau(W)\bm{b}(W)]$.
Under weak regularity conditions, this sample-splitting-based semiparametric estimator attains the efficiency bound with the double robustness property, and thus allow us to employ flexible machine learning methods while achieving the $\sqrt{n}$-rate of convergence and valid inference \citep{kennedy2017semiparametric}\footnote{If one is willing to rely on appropriate empirical process conditions (e.g., Donsker-type or low entropy conditions \citep{van2000asymptotic}), then $\eta$ can be estimated on the same sample without sample splitting. However this would limit the flexibility of the nuisance estimators.}. 

Development of an efficient estimator for $\E\left\{\text{uf}_j(Z)\bm{b}(W) \right\}$ depends on the form of the fairness function $\text{uf}_j$. Here, we provide a few illustrative cases. First, for the fairness function \eqref{eqn:g-independence} corresponding to the criterion of independence, one may simply use the following sample-average type estimator:
\begin{align*}
    \Pn\left\{\widehat{\text{uf}}_j\bm{b}(W) \right\} = \Pn\left[\left\{\frac{(1-S)}{\Pn(1-S)} - \frac{S}{\Pn(S)}\right\}\bm{b}(W)\right],
\end{align*}
which is naturally $\sqrt{n}$-consistent without any need for nuisance estimation.

Next, for the fairness functions involving a non-smooth component as in \eqref{eqn:g-positive-balance}, we can employ some standard techniques in the nonparametric literature. One of them is the margin condition. In case of \eqref{eqn:g-positive-balance}, we can restrict the probability that the two outcome regression functions get too close to each other, by imposing the following margin condition:
\begin{definition}[Margin Condition] \label{assumption:MC}
For any margin exponent $\alpha > 0$ and for all $t \geq 0$,
\begin{align} \label{eqn:margin-condition}
    \Pb(\vert \mu_1(W) - \mu_0(W)  \vert \leq t) \lesssim t^{\alpha}.    
\end{align}
\end{definition}
The above margin condition is analogous to that used in \citet{kitagawa2018should, luedtke2016optimal, luedtke2016statistical, kennedy2019survivor, kennedy2020sharp} as well as other problems involving estimation of non-smooth functionals such as classification \citep{audibert2007fast}, couneterfactual density estimation \citep{kim2018causal}, and clustering \citep{levrard2018quantization}. This margin condition allows us to use the following plug-in type estimator with sample splitting:
\begin{align*}
    & \Pn\left\{\widehat{\text{uf}}_j\bm{b}(W) \right\} = \frac{1}{K}\sum_{b=1}^K\Pn^b\Bigg(\bm{b}(W)\times \\
    &\quad \Bigg[\frac{(1-S)\mathbbm{1}\left\{\widehat{\mu}_{1,-b}(W) - \widehat{\mu}_{0,-b}(W)>0\right\} }{\frac{1}{K}\sum_{b=1}^K\Pn^b\left[(1-S)\mathbbm{1}\left\{\widehat{\mu}_{1,-b}(W) - \widehat{\mu}_{0,-b}(W)>0\right\}\right]} \notag \\
    & \quad - \frac{S\mathbbm{1}\left\{\widehat{\mu}_{1,-b}(W) - \widehat{\mu}_{0,-b}(W)>0\right\}}{\frac{1}{K}\sum_{b=1}^K\Pn^b\left[S \mathbbm{1}\left\{\widehat{\mu}_{1,-b}(W) - \widehat{\mu}_{0,-b}(W)>0\right\}\right]}\Bigg]\Bigg).
\end{align*}
The above estimator attains $\sqrt{n}$-consistency and asymptotic normality if $\Pb\left[ \mathbbm{1}\{\widehat{\mu}_1-\widehat{\mu}_0 > 0\} \neq \mathbbm{1}\{\mu_{1}-\mu_{0} > 0\}\right]=o_\Pb(1)$ and $\max_a \Vert \widehat{\mu}_{a}- \mu_{a} \Vert_{\infty, \Pb}^{\alpha}=o_\Pb(n^{-\frac{1}{2}})$ (more details can be found in Appendix \ref{app:margin-condition}). Alternatively, one may use  smoothing approximation techniques such as kernel or polynomial smoothing \citep[e.g.,][]{kim2018causal, kennedy2020optimal}. We do not discuss each of them here as it is beyond the scope of this paper. 

Further, if our fairness function involves the potential outcomes, we can utilize estimators similar to those employed for our objective function. Specifically, if the fairness function $\text{uf}_j(Y^0,Y^1,W)$ is a smooth function of $Y^1, Y^0$, then one may use the following EIF-based semiparametric estimator
\begin{align*}     
    & \Pn\left\{\widehat{\text{uf}}_j\bm{b}(W) \right\} \\
    & = \Pn\left\{\text{uf}_j\left(\varphi_0(Z;\widehat{\eta}_{-B}),\varphi_1(Z;\widehat{\eta}_{-B}),W\right)\bm{b}(W) \right\},
\end{align*}
which is asymptotically normal and efficient according to the same logic as used for \eqref{eqn:estimator-objective}.

Consequently, our approximating program can be found as the following convex quadratic program (QP):
\begin{equation}
\label{eqn:approx-opt-problem}
\begin{aligned}
    & \underset{\beta \in \R^k}{\text{minimize}} \quad  \frac{1}{2}\beta^\top\Pn\left\{\bm{b}(W)\bm{b}(W)^\top\right\} \beta \\
    & \phantom{\underset{\beta \in \R^k}{\text{minimize}}} \quad - \beta^\top\Pn\left[\left\{ \varphi_1(Z;\widehat{\eta}_{-B}) - \varphi_0(Z;\widehat{\eta}_{-B})\right\} \bm{b}(W) \right] \\
    & \text{subject to } \quad \left\vert \beta^\top\Pn\left\{\widehat{\text{uf}}_j\bm{b}(W) \right\} \right\vert \leq \delta_{j}, \, j \in J.
\end{aligned}     \tag{$\widehat{\mathsf{P}}$}  
\end{equation}
The above QP can be readily solved using off-the-shelf solvers. Let $\widehat{\beta}$ be an optimal solution to \eqref{eqn:approx-opt-problem}. Our proposed estimator for $\tau$ is then given by $\widehat{\tau}(W) = \widehat{\beta}^\top\bm{b}(W)$.


Next, we introduce the following assumptions on our basis expansion and nuisance estimation. 

\begin{enumerate}[label={}, leftmargin=*]
    \item \customlabel{assumption:A1}{(A1)} $\E[\bm{b}(W)\bm{b}(W)^\top]$ is positive definite 
    \item \customlabel{assumption:A2}{(A2)}  $\Pb(\widehat{\pi}_a \in [\epsilon, 1-\epsilon]) = 1$ for some $\epsilon > 0$ 
    \item \customlabel{assumption:A3}{(A3)} $\Vert \widehat{\mu}_{a}- \mu_{a} \Vert_{2,\Pb} = o_\Pb(1)$ or $\Vert \widehat{\pi}_{a}- \pi_{a} \Vert_{2,\Pb} = o_\Pb(1)$
    \item \customlabel{assumption:A4}{(A4)}  
    $   \begin{aligned}[t]
            &\Vert \widehat{\pi}_{a} - {\pi}_{a} \Vert_{2,\Pb} \Vert \widehat{\mu}_{a} - \mu_{a} \Vert_{2,\Pb} = o_\Pb(n^{-\frac{1}{2}})
            \end{aligned} 
        $
\end{enumerate}	  
Assumption \ref{assumption:A1}, requiring that our basis functions are never perfectly collinear, ensures that the quadratic growth condition holds at the optimal solution in \eqref{eqn:true-opt-problem}. This could be replaced with a weaker technical condition \citep[see][or Appendix \ref{appsec:definitions-LICQ-SC}]{kim2022counterfactualmv}. It also guarantees uniqueness of the optimal solution to \eqref{eqn:true-opt-problem}. Assumptions \ref{assumption:A2} - \ref{assumption:A4} are conditions regarding our nuisance estimation, and commonly used in semiparametric estimation in the causal inference literature \citep[e.g.,][]{kennedy2016semiparametric, kennedy2022semiparametric}. We also need the following assumptions to ensure that the stochastic components in the constraints are consistent at fast rates and attain asymptotic normality.
\begin{enumerate}[label={}, leftmargin=*]
    \item \customlabel{assumption:A5}{(A5)} $\Pn\left\{\widehat{\text{uf}}_j\bm{b}(W) \right\} - \E\left\{\text{uf}_j\bm{b}(W)\right\}=O_\Pb\left(n^{-\frac{1}{2}}\right)$ 
    \item \customlabel{assumption:A6}{(A6)} $\sqrt{n}\left[\Pn\left\{\widehat{\text{uf}}_j\bm{b}(W) \right\} - \E\left\{\text{uf}_j\bm{b}(W)\right\}\right]$ converges in distribution to a normal random variable with zero mean and finite variance.
\end{enumerate}	 

In the following theorem, we provide the large-sample properties of our proposed estimator.

\begin{theorem} \label{thm:asymptotics}
Let $\beta^*$ and $\widehat{\beta}$ denote the optimal solutions to \eqref{eqn:true-opt-problem} and \eqref{eqn:approx-opt-problem}, respectively. If Assumptions \ref{assumption:A1} - \ref{assumption:A3}, and \ref{assumption:A5} hold, then
\begin{align*}
    \Vert \widehat{\beta} - \beta^* \Vert_2 = O_\Pb\left(\max_a \Vert \widehat{\pi}_a -\pi_a \Vert_{2,\Pb} \Vert \widehat{\mu}_a -\mu_a \Vert_{2,\Pb} + n^{-\frac{1}{2}}\right).
\end{align*}
If we additionally assume \ref{assumption:A4} and \ref{assumption:A6}, and that the Linear Independence Constraint Qualification (LICQ) and Strict Complementarity (SC) hold at $\beta^*$, then $\sqrt{n}(\widehat{\beta} - \beta^*)$ converges in distribution to a normal random variable with zero mean and finite variance. 
\end{theorem}

The above result immediately follows by Theorems 3.1 and 3.2 of \citet{kim2022counterfactualmv}, and gives conditions under which $\widehat{\beta}$ is $\sqrt{n}$-consistent and asymptotically normal. Thus, valid confidence intervals and hypothesis tests can be constructed via the bootstrap. The doubly robust second-order term $\Vert \widehat{\pi}_a -\pi_a \Vert_{2,\Pb} \Vert \widehat{\mu}_a -\mu_a \Vert_{2,\Pb}$ will be small if either $\pi_a$ or $\mu_a$ are estimated accurately. In nonparametric modeling, this second-order error substantially lowers the bar for the nuisance estimator convergence rate, which allows much more flexible methods to be employed  while still achieving $\sqrt{n}$ rates; for example, it suffices that both nuisance functions are estimated consistently at $n^{\frac{1}{4}}$ rates \citep{kennedy2016semiparametric}. LICQ and SC are regularity conditions commonly found in the optimization literature \citep[e.g.,][]{still2018lectures, shapiro2014lectures}; see Appendix \ref{appsec:definitions-LICQ-SC} for the formal definitions. Here we use weaker assumptions than the standard results in stochastic optimization \citep{shapiro1993asymptotic} due to the linear constraints \citep{kim2022counterfactualmv}. We also remark that, while we choose the mean squared error $\mathcal{R}_{\text{MSE}}$ as our default risk function, other risk functions, even those with regularization terms, may be used in our framework if they satisfy the necessary regularity conditions \citep[see][Section 3]{kim2022counterfactualmv}.


\section{Regret Bounds and Fairness-Welfare Tradeoff} \label{sec:regret-and-tradeoff}

In this section, we analyze the regret upper bounds and discuss the implications of using the proposed CATE estimator under fairness criteria for the estimation of the optimal policy $g^*$ defined in \eqref{eqn:optimal-trt-regime}. One may estimate $g^*$ by
\begin{align} \label{eqn:policy-estimator}
    \widehat{g}(W) = \mathbbm{1}\left\{\widehat{\beta}^\top\bm{b}(W) > 0\right\},
\end{align}
where $\widehat{\beta}^\top\bm{b}$ is our estimates of $\tau$ obtained by solving \eqref{eqn:approx-opt-problem}. Note that since $\widehat{g}$ is a non-smooth function of $\widehat{\beta}^\top\bm{b}$, one should not in general expect to achieve the same level of fairness for $\widehat{g}$ as for $\widehat{\beta}^\top\bm{b}$ (see Appendix \ref{appsec:nonlinearity-example} for more details). Following the convention in the literature, we evaluate the performance of the above estimated policy in terms of the welfare loss or \emph{regret} relative to the maximum achievable welfare $\mathcal{U}(g^*)$, i.e., $\mathcal{U}(g^*) - \mathcal{U}(\widehat{g})$. 

We shall use the margin condition \eqref{assumption:MC} for our analysis, which helps us to avoid expensive classification errors for $\widehat{g}$ when $\tau$ is close to $0$. In the following lemma, we adapt the comparison results in \citet{audibert2007fast} and provide two useful inequalities between the regrets and the $L_q$ risks of the proposed CATE estimator.

\begin{lemma} \label{lem:comparison-inequalities}
    Assume that the margin condition \eqref{eqn:margin-condition} holds with margin exponent $0 < \alpha <\infty$.
    Define the risk score  $\Delta \equiv \Delta(W) = \tau(W)$ and $\widehat{\Delta} \equiv \widehat{\Delta}(W) = \widehat{\beta}^\top\bm{b}(W)$. 
    Then 
    \begin{align*}
        \mathcal{U}(g^*) - \mathcal{U}(\widehat{g}) \lesssim  \left\Vert   \widehat{\Delta} - \Delta \right\Vert_{\infty,\Pb}^{\alpha+1}.
    \end{align*}
    Further, for any $1 \leq q < \infty$, 
    \begin{align*}
        \mathcal{U}(g^*) - \mathcal{U}(\widehat{g}) \lesssim  \left\Vert   \widehat{\Delta} - \Delta \right\Vert_{q,\Pb}^{\frac{q(1+\alpha)}{q+\alpha}}.
    \end{align*}
\end{lemma}

Using this lemma, the following theorem provides upper bounds of the regret for our proposed estimator $\widehat{g}$ in \eqref{eqn:policy-estimator}. 
These results are asymptotic in the sample size $n$.

\begin{theorem} \label{thm:regret-bound}
Assume \ref{assumption:A1} - \ref{assumption:A3}, \ref{assumption:A5} and that the margin condition \eqref{eqn:margin-condition} holds with margin exponent $0 < \alpha <\infty$. 
Let
\begin{equation} \label{eqn:true-unconstrained}
    \widetilde{\beta} =  \underset{\beta \in \R^k}{\argmin} \quad \E\left\{ \left( Y^1-Y^0 - \beta^\top\bm{b}(W) \right)^2\right\},
\end{equation}
and define the remainder terms 
\begin{align*}
     & R_{1,n} = O_\Pb\left(\max_a \Vert \widehat{\pi}_a -\pi_a \Vert_{2,\Pb} \Vert \widehat{\mu}_a -\mu_a \Vert_{2,\Pb} + n^{-\frac{1}{2}}\right), \\
     & R_{2}= O\Big(\Big\Vert \sum_j \sqrt{\lambda_j} \text{uf}_j(Z)\Big\Vert_{2,\Pb} \Big\Vert \bm{b}(W) \Big\Vert_{2,\Pb} \Big),
\end{align*} 
where $\lambda_j \geq 0$ is the Lagrange multiplier associated with the $j$-th fairness constraint in \eqref{eqn:true-opt-problem}.
Then 
\begin{enumerate}
    \item[(i)] 
    $   \begin{aligned}[t]
            \mathcal{U}(g^*) - \mathcal{U}(\widehat{g}) \lesssim & \left\Vert \tau(W) - {\widetilde{\beta}}^\top\bm{b}(W) \right\Vert_{\infty, \Pb}^{1+\alpha} \\
            & + R_{1,n}^{1+\alpha} + R_{2}^{1+\alpha}, 
            \end{aligned} 
    $
    \item[(ii)]
    $   \begin{aligned}[t]
            \Pb\left\{\widehat{g}(W) \neq {g^*}(W) \right\} \lesssim & \left\Vert \tau(W) - {\widetilde{\beta}}^\top\bm{b}(W) \right\Vert_{\infty, \Pb}^{\alpha} \\
            & + R_{1,n}^{\alpha} + R_{2}^{\alpha}, 
            \end{aligned} 
    $
    \item[(iii)]
    $   \begin{aligned}[t]
            \mathcal{U}(g^*) - \mathcal{U}(\widehat{g}) \lesssim & \left\Vert \tau(W) - {\widetilde{\beta}}^\top\bm{b}(W) \right\Vert_{q,\Pb}^{\frac{q(1+\alpha)}{q+\alpha}} \\
            & + R_{1,n}^{\frac{q(1+\alpha)}{q+\alpha}} + R_{2}^{\frac{q(1+\alpha)}{q+\alpha}}, \forall  1 \leq q < \infty.
            \end{aligned} 
    $
\end{enumerate}

\end{theorem}

In (ii), $\Pr\left\{\widehat{g}(W) \neq {g^*}(W) \right\}$ is the probability that $\widehat{g}$ differs from the true optimal policy $g^*$ over a new observation. Theorem \ref{thm:regret-bound} shows that the regret bounds depend on the levels of both the nuisance estimation accuracy and fairness we wish to achieve in the CATE estimation problem.

Specifically, each bound listed in Theorem \ref{thm:regret-bound} consists of three terms. The first term is an unavoidable modeling error minimized through least square estimation, which will vanish if $\mu_a(\cdot)$ lies in the function space spanned by the basis functions $\bm{b}(\cdot)$. 

The second term, $R_{1,n}$, is essentially the doubly robust second-order term that appears in Theorem \ref{thm:asymptotics}. $R_{1,n}$ converges to zero at $\sqrt{n}$ rates even when $\pi_a$ and $\mu_a$ are flexibly estimated at slower than $\sqrt{n}$ rates.

The third term, $R_2$, has important implications. 
It indicates the cumulative unfairness in the CATE with respect to the sensitive features, measured by a series of the fairness functions. If we use small values of the tolerance level $\delta_{j}$ so that the optimum $\beta^*$ is constrained by the $j$-th fairness constraint (i.e., the $j$-th constraint is active), then the corresponding Lagrange multiplier, $\lambda_j$, is positive. On the contrary, if we loosen the standard by using large values of $\delta_{j}$ so that the $j$-th fairness constraint does not constrain $\beta^*$, $\lambda_j$ is set to zero. Therefore, our attempts toward making optimal policies more fair may lead to an additional welfare loss (regret) relative to the universally maximum possible welfare $\mathcal{U}(g^*)$. In other words, there is a trade-off between fairness in the optimal policy and the maximum achievable welfare.

In short, Theorem \ref{thm:regret-bound} states that, while the proposed approach can significantly reduce the cost of estimating the nuisance components, there is still a cost associated with the fairness constraints in order to estimate the optimal policy at the required level of fairness.

Finally, we remark that a comparable analysis for $g_\mathcal{I}$ in \eqref{eqn:subgroup-rule} can also be performed with a slightly different margin condition, where the probability is restricted on the edge of the interval $\mathcal{I}$.



\section{Experiments}
\subsection{Simulation study} \label{sec:simulation}

\begin{figure*}[!t]
\centering
\begin{minipage}[t]{.57\linewidth}
  \centering
  \includegraphics[width=\linewidth]{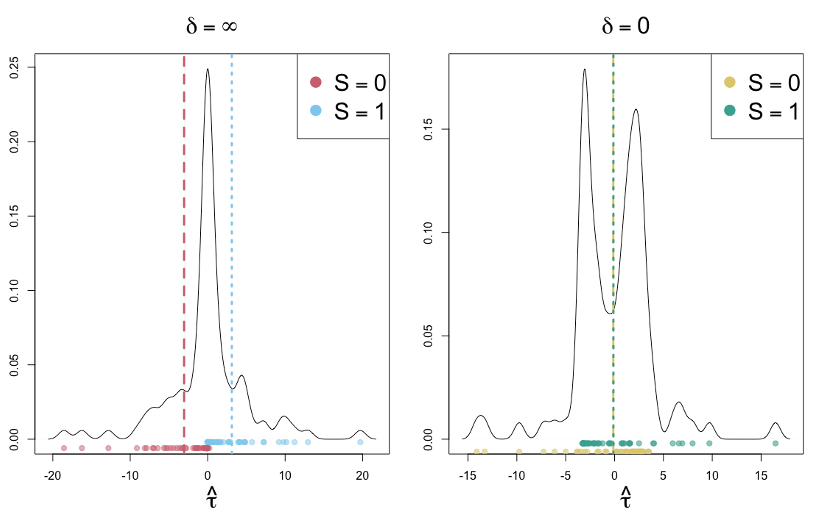}
  \captionof{figure}{Densities of $\widehat{\tau}(W)$ for $\delta=\infty$ and $\delta=0$. In each figure, two dashed vertical lines correspond to $\Pn\left\{\widehat{\tau} \mid S=0\right\}$ and $\Pn\left\{\widehat{\tau} \mid S=1\right\}$.}
  \label{fig:tauhat-uf}
\end{minipage}%
\hfill
\begin{minipage}[t]{.33\linewidth}
  \centering
  \includegraphics[width=\linewidth]{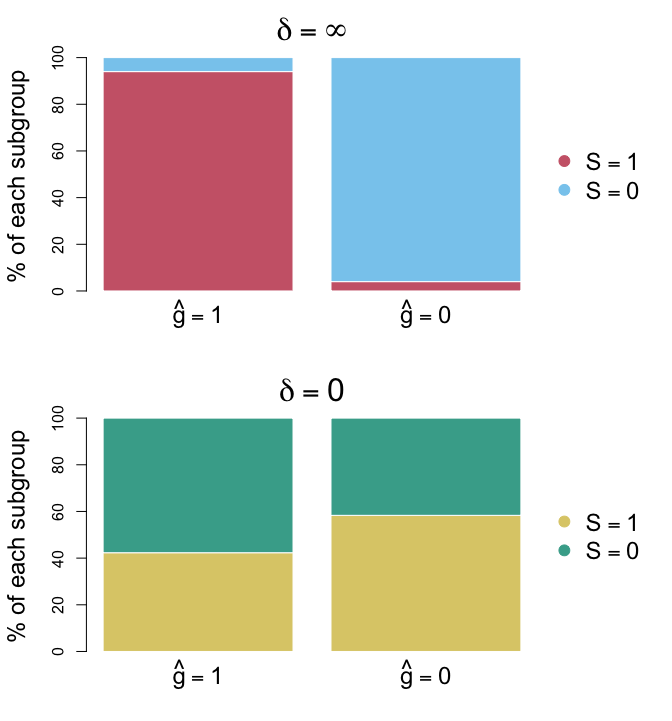}
  \captionof{figure}{Proportions of $S=1$ to $S=0$ for each decision of $\widehat{g}$.}
  \label{fig:dist-S}
\end{minipage}
\end{figure*}

In this section, we conduct a simulation study to illustrate the theoretical properties and finite-sample performance of the proposed estimators. We generate synthetic data according to the process described in Section \ref{sec:motivating-example} where we set 
\begin{align*}
& f_\mu(W) = \log(SX_1^2+10) + \exp(-SX_2/5) + SX_1, \\    
& Y^a = \mu_a(W) + \epsilon, \, \epsilon \sim N(0,1), \quad Y = Y^1 A + Y^0 (1-A).
\end{align*}
In the following experiments, we consider the independence criterion with the fairness function defined in \eqref{eqn:g-independence}. $\bm{b}(W)$ consists of the polynomial terms $W, W^2, W^3$ and $\{W_jW_kW_s\}_{j,k,s}$ to form the third-order Taylor expansion. All nuisance functions are estimated using the cross-validation super learner ensemble implemented in the \texttt{SuperLearner} R package to combine generalized additive models, adaptive regression splines, and random forests.

We first assess the ability of the proposed estimators to mitigate unfairness in the estimated CATE. To this end, we generate a sample of size $2000$ observations and estimate $\widehat{\tau}$ using \eqref{eqn:approx-opt-problem} with $K=2$ splits, under $\delta=\infty$ (i.e., with no fairness constraints) and $\delta=0$. 
We then compare the densities of $\widehat{\tau}$ for each value of $\delta$, and show how individuals belonging to distinct groups are distributed in terms of $S$ in Figure \eqref{fig:tauhat-uf}. In Figure \ref{fig:dist-S}, we compute the proportions of $S=1$ to $S=0$ for each decision of $\widehat{g}(W) = \mathbbm{1}\left\{\widehat{\tau}(W) > 0\right\}$ under $\delta=\infty$ and $\delta=0$. 

In Figure \ref{fig:tauhat-uf}, without the fairness constraints, we observe a pronounced violation of the independence criterion in $\widehat{\tau}$; individuals belonging to the group $S=1$ ($S=0$) are mostly distributed in the regime $\widehat{\tau}>0$ ($\widehat{\tau}<0$). This eventually leads to disproportionate (unfair) policies (top of Figure \ref{fig:dist-S}) where individuals with $S=1$ are to be treated.
However, when the fairness constraint is applied with $\delta=0$, this issue is largely resolved. In Figure \ref{fig:tauhat-uf}, we observe that the conditional sample means of $\widehat{\tau}$ given $S=0$ and $S=1$ are nearly identical, and many individuals belonging to the $S=1$ ($S=0$) group are shifted to the left (right) compared to the case $\delta=\infty$. This produces policies where the treated and untreated groups are more balanced along the sensitive feature, giving individuals with $S=0$ more chances to be treated (bottom of Figure \ref{fig:dist-S}). 

To illustrate the theoretical findings in Section \ref{sec:regret-and-tradeoff}, which state that the fairness in optimal policies comes at a price, we estimate $g^*$ across different values of $\delta$ ranging between 0 and 4. 
In our simulation settings, $\delta=4$ is equivalent to having no constraints ($\delta=\infty$) as the fairness constraint becomes inactive. 
This time, we also estimate the welfare $\mathcal{U}(g^*)$ and the unfairness of our policy $\E\left\{{\text{uf}}(Z){g^*}(W) \right\}$ on a separate independent sample of the same size. 
We compare our methods to the doubly robust score approach (WA21) of \citet{wager2018estimation}, and the m-hybrid (KT18-m) and e-hybrid (KT18-e) rules by \citet{kitagawa2018should}. 
We do not incorporate our fairness constraint for these approaches since it is not addressed in the original works. The results across $500$ simulations are presented in Figure \ref{fig:util-uf}.

\begin{figure}[t!]
\centering
    \includegraphics[width=.975\linewidth]{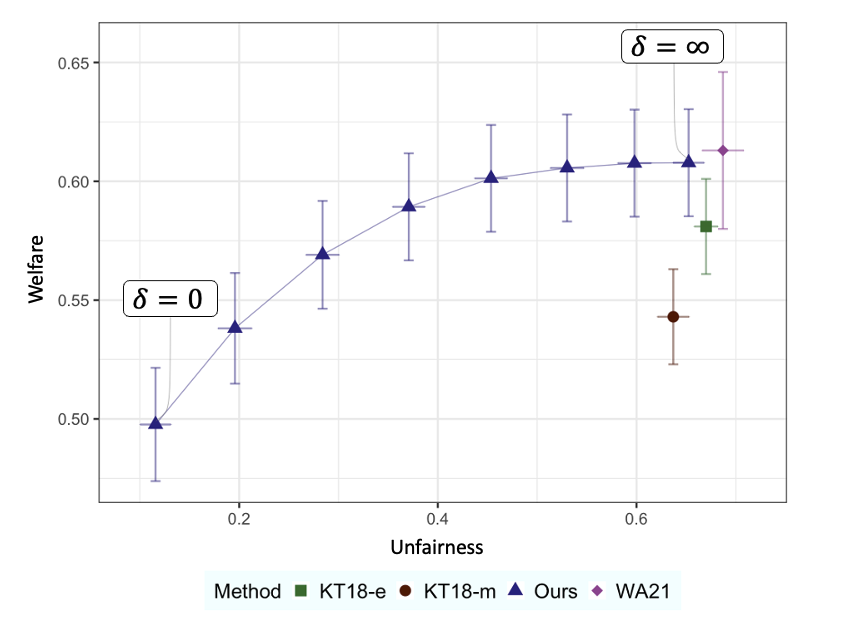}
    \caption{The curve visualizes how much welfare is required to sacrifice in order to achieve a desired level of fairness.} \label{fig:util-uf}
    \vspace*{-.1in}
\end{figure}

\begin{figure*}[!t]
\centering
\begin{minipage}{.45\linewidth}
  \centering
  \includegraphics[width=\linewidth]{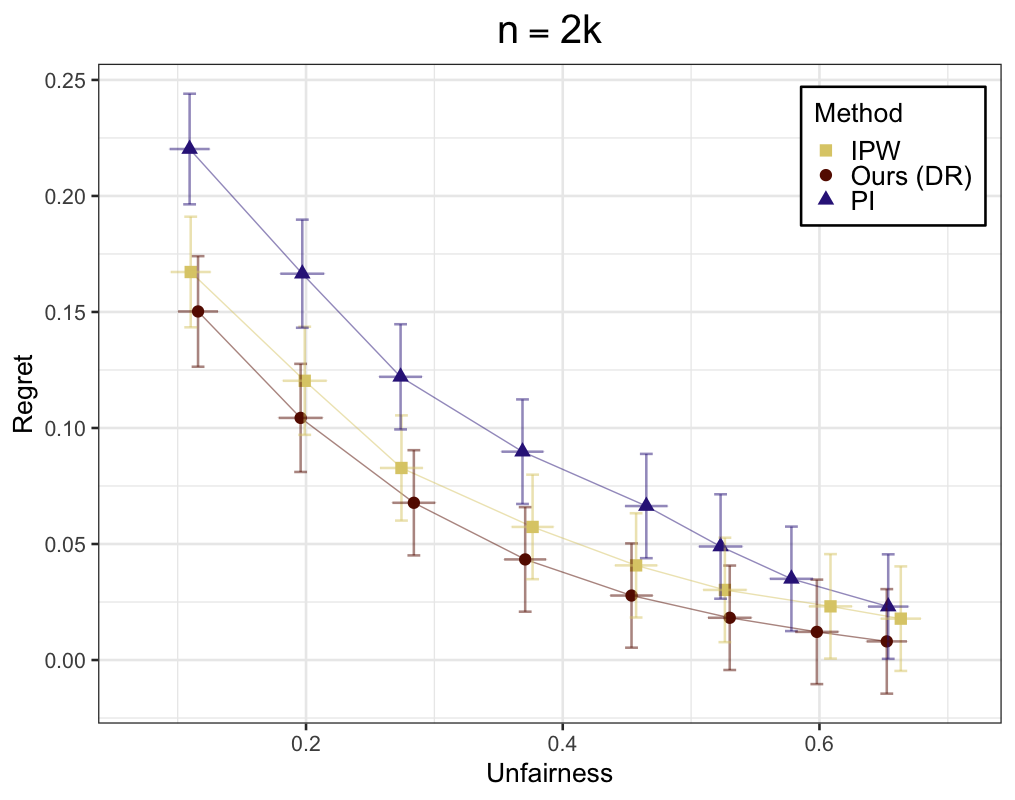}
\end{minipage}
\hfill
\begin{minipage}{.45\linewidth}
  \centering
  \includegraphics[width=\linewidth]{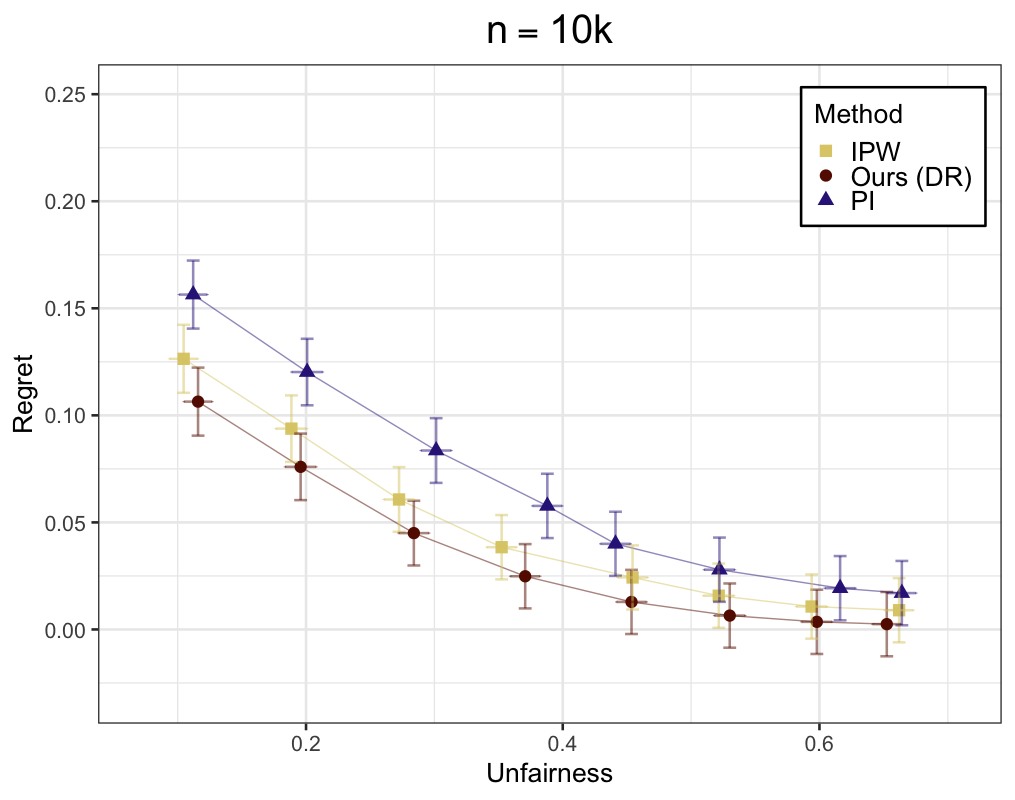}
\end{minipage}
\caption{Regret and unfairness curves of the three estimators for $n=2k$ (left) and $n=10k$ (right).} \label{fig:regret-uf}
\end{figure*}

The results illustrate the predictions of our theory. 
In Figure \ref{fig:util-uf}, the estimated welfare gradually decreases as we apply stricter tolerance levels ($\delta$). The figure also shows that one may reduce unfairness without any noticeable decrease in welfare up to a certain point.
To get the result of the bottom of Figure \eqref{fig:dist-S}, we have to sacrifice the overall welfare by $0.13$. 
This illustrates the trade-off between fairness and social welfare. Further in the simulation, we see that the cost of reducing the same amount of unfairness grows rapidly as we continue to decrease $\delta$.
Without the fairness constraints, the proposed method achieves a comparable performance to the doubly robust approach of \citet{athey2021policy}.

Finally, to illustrate the double robustness property in Theorems \ref{thm:asymptotics} and \ref{thm:regret-bound}, we consider three estimators: the proposed estimator and the previous two where the counterfactual component is estimated by the PI and IPW estimators. 
In each simulated dataset, again $\widehat{g}$ is estimated on samples of size $2000$ and $10000$ across the same values of $\delta$ as above, and the regret and the unfairness are estimated on new, separate samples. 
The results across 500 simulations are presented in Figure \ref{fig:regret-uf}. 
Again we observe the trade-off between fairness and regret. 
However, the regret of the proposed method decays faster than the other methods as the sample size grows, due to the double robustness property (the second-order product of nuisance errors).

\subsection{Case study}
We next illustrate our methods on the COMPAS dataset originally gathered to assess the risk of recidivism \citep{angwin2016machine}. 
Following \citet{mishler2021fairness}, we let $A$ represent pretrial release, with $A=1$ if defendants are released and $A=0$ if they are incarcerated. Our sensitive feature represents race, restricted to defendants who are Caucasian ($n = 2013$) and African-American ($n = 3175$), coded as $S = 0$ and $S = 1$, respectively. The outcome of interest, $Y$, is rearrest within two years. We use three covariates: age, sex, and number of prior arrests. 
We are interested in estimating the policy that minimizes the risk of recidivism, so in this case our target optimal policy is $g^*(W) = \mathbbm{1}\left\{\mu_1(W) < \mu_0(W)\right\}$ where the corresponding welfare defines the risk of recidivism. 
We use roughly two-thirds of the data to estimate $\widehat{g}$, and the rest to estimate the welfare and unfairness using the same setting as in the preceding subsection. Here, we consider two fairness criteria simultaneously: independence \eqref{eqn:g-independence} (IDP) and positive balance \eqref{eqn:g-positive-balance} (PB).

In Table \ref{tbl:compas-case}, we compare our estimators with $\delta=0$ and $\delta=\infty$ to the three other approaches used in Figure \ref{fig:util-uf}. 
As expected, without the fairness constraint, while producing similar risks, all the methods show large disparities; in fact, the estimated optimal policies suggest incarcerating more African-American defendants. 
However, when the proposed estimator is used with $\delta=0$, both disparities are jointly reduced by a significant margin with a small impact on the risk. This illustrates the effectiveness of the proposed methodology.

\begin{table}[t!]
\centering
    \begin{tabular}{cccc}\hline
        Method & \thead{Recidivism \\ Risk} & IDP & PB \\ \hline \hline
        KT18-e & 0.33 & 0.28 & 0.15 \\
        KT18-m & {0.28} & 0.31 & 0.16  \\
        WA21 & {0.25} & {0.33} & 0.17 \\
        Ours ($\delta=\infty$) & {0.27} & {0.32} & 0.16  \\
        Ours ($\delta=0$) & 0.34 & \textbf{0.09} & \textbf{0.07} \\ \hline
     \end{tabular}
     \caption{The estimated risk of recidivism and unfairness.} \label{tbl:compas-case}
     \vspace*{-.1in}
\end{table}

\section{Discussion}

We propose a new framework for fair and robust estimation of heterogeneous treatment effects and discuss its implications for policy learning. Our method is easily implementable and allows practitioners to flexibly incorporate various fairness constraints to meet the desired level of fairness. This affords new opportunities to leverage recent advances in algorithmic fairness for robust estimation of heterogeneous treatment effects. The proposed CATE estimator enables us to design targeting policies in a fairer fashion. Importantly, our theoretical results in Section \ref{sec:regret-and-tradeoff} help us to understand the trade-off between fairness and welfare. 
Although this study focused on univariate and binary sensitive variables, our proposed framework is easily extendable to a multivariate sensitive variable with discrete support.

There are some caveats to mention, and ways in which our work could be improved. Our method is intended to attain a desired level of fairness in CATE estimation, not in targeting policies. Different methods are required to estimate a particular targeting policy that satisfies a specific fairness criterion. In a forthcoming paper, we develop a novel nonparametric estimator for this kind of fair optimal policies. Further, although the setup we consider here is widely used, in future work it is desirable to consider extensions to time-varying treatments, instrumental variables, and mediation. 
Finally, it is also of interest to apply the proposed methods to various real-world problems, to bring new insights into fair heterogeneous treatment effect estimation and policy learning.

\section{Acknowledgments}
This work was completed while Kwangho Kim was a research associate at Harvard Medical School. This work was supported by the National Research Foundation of Korea (NRF) grant funded by the Korea governement (MSIT)(No. NRF-2022M3J6A1063595).


\bibliography{bibliography}

\begin{thebibliography}{70}
\providecommand{\natexlab}[1]{#1}
\providecommand{\url}[1]{\texttt{#1}}
\expandafter\ifx\csname urlstyle\endcsname\relax
  \providecommand{\doi}[1]{doi: #1}\else
  \providecommand{\doi}{doi: \begingroup \urlstyle{rm}\Url}\fi

\bibitem[Angwin et~al.(2016)Angwin, Larson, Mattu, and
  Kirchner]{angwin2016machine}
Angwin, J., Larson, J., Mattu, S., and Kirchner, L.
\newblock Machine bias.
\newblock In \emph{Ethics of Data and Analytics}, pp.\  254--264. Auerbach
  Publications, 2016.

\bibitem[Athey \& Imbens(2016)Athey and Imbens]{athey2016recursive}
Athey, S. and Imbens, G.
\newblock Recursive partitioning for heterogeneous causal effects.
\newblock \emph{Proceedings of the National Academy of Sciences}, 113\penalty0
  (27):\penalty0 7353--7360, 2016.

\bibitem[Athey \& Wager(2021)Athey and Wager]{athey2021policy}
Athey, S. and Wager, S.
\newblock Policy learning with observational data.
\newblock \emph{Econometrica}, 89\penalty0 (1):\penalty0 133--161, 2021.

\bibitem[Audibert \& Tsybakov(2007)Audibert and Tsybakov]{audibert2007fast}
Audibert, J.-Y. and Tsybakov, A.~B.
\newblock Fast learning rates for plug-in classifiers.
\newblock \emph{The Annals of statistics}, 35\penalty0 (2):\penalty0 608--633,
  2007.

\bibitem[Ballarini et~al.(2018)Ballarini, Rosenkranz, Jaki, K{\"o}nig, and
  Posch]{ballarini2018subgroup}
Ballarini, N.~M., Rosenkranz, G.~K., Jaki, T., K{\"o}nig, F., and Posch, M.
\newblock Subgroup identification in clinical trials via the predicted
  individual treatment effect.
\newblock \emph{PloS one}, 13\penalty0 (10):\penalty0 e0205971, 2018.

\bibitem[Barocas et~al.(2019)Barocas, Hardt, and
  Narayanan]{barocas-hardt-narayanan}
Barocas, S., Hardt, M., and Narayanan, A.
\newblock \emph{Fairness and Machine Learning}.
\newblock fairmlbook.org, 2019.
\newblock \url{http://www.fairmlbook.org}.

\bibitem[Brinkley et~al.(2010)Brinkley, Tsiatis, and
  Anstrom]{brinkley2010generalized}
Brinkley, J., Tsiatis, A., and Anstrom, K.~J.
\newblock A generalized estimator of the attributable benefit of an optimal
  treatment regime.
\newblock \emph{Biometrics}, 66\penalty0 (2):\penalty0 512--522, 2010.

\bibitem[Chen et~al.(2017)Chen, Tian, Cai, and Yu]{chen2017general}
Chen, S., Tian, L., Cai, T., and Yu, M.
\newblock A general statistical framework for subgroup identification and
  comparative treatment scoring.
\newblock \emph{Biometrics}, 73\penalty0 (4):\penalty0 1199--1209, 2017.

\bibitem[Chernozhukov et~al.(2017)Chernozhukov, Chetverikov, Demirer, Duflo,
  Hansen, and Newey]{Chernozhukov17}
Chernozhukov, V., Chetverikov, D., Demirer, M., Duflo, E., Hansen, C., and
  Newey, W.
\newblock Double/debiased/neyman machine learning of treatment effects.
\newblock \emph{American Economic Review}, 107\penalty0 (5):\penalty0 261--65,
  May 2017.

\bibitem[Corbett-Davies et~al.(2017)Corbett-Davies, Pierson, Feller, Goel, and
  Huq]{corbett2017algorithmic}
Corbett-Davies, S., Pierson, E., Feller, A., Goel, S., and Huq, A.
\newblock Algorithmic decision making and the cost of fairness.
\newblock In \emph{Proceedings of the 23rd acm sigkdd international conference
  on knowledge discovery and data mining}, pp.\  797--806, 2017.

\bibitem[Coston et~al.(2020{\natexlab{a}})Coston, Kennedy, and
  Chouldechova]{coston2020RuntimeConfounding}
Coston, A., Kennedy, E., and Chouldechova, A.
\newblock Counterfactual predictions under runtime confounding.
\newblock In \emph{Advances in neural information processing systems},
  volume~33, pp.\  4150--4162, 2020{\natexlab{a}}.

\bibitem[Coston et~al.(2020{\natexlab{b}})Coston, Mishler, Kennedy, and
  Chouldechova]{coston2020counterfactual}
Coston, A., Mishler, A., Kennedy, E.~H., and Chouldechova, A.
\newblock Counterfactual risk assessments, evaluation, and fairness.
\newblock In \emph{Proceedings of the 2020 Conference on Fairness,
  Accountability, and Transparency}, pp.\  582--593, 2020{\natexlab{b}}.

\bibitem[Emanuel et~al.(2020)Emanuel, Persad, Upshur, Thome, Parker, Glickman,
  Zhang, Boyle, Smith, and Phillips]{emanuel2020fair}
Emanuel, E.~J., Persad, G., Upshur, R., Thome, B., Parker, M., Glickman, A.,
  Zhang, C., Boyle, C., Smith, M., and Phillips, J.~P.
\newblock Fair allocation of scarce medical resources in the time of covid-19,
  2020.

\bibitem[Gelman et~al.(2007)Gelman, Fagan, and Kiss]{gelman2007analysis}
Gelman, A., Fagan, J., and Kiss, A.
\newblock An analysis of the new york city police department's
  “stop-and-frisk” policy in the context of claims of racial bias.
\newblock \emph{Journal of the American statistical association}, 102\penalty0
  (479):\penalty0 813--823, 2007.

\bibitem[Hardt et~al.(2016)Hardt, Price, and Srebro]{hardt2016equality}
Hardt, M., Price, E., and Srebro, N.
\newblock Equality of opportunity in supervised learning.
\newblock \emph{Advances in neural information processing systems}, 29, 2016.

\bibitem[Henderson et~al.(2010)Henderson, Ansell, and
  Alshibani]{henderson2010regret}
Henderson, R., Ansell, P., and Alshibani, D.
\newblock Regret-regression for optimal dynamic treatment regimes.
\newblock \emph{Biometrics}, 66\penalty0 (4):\penalty0 1192--1201, 2010.

\bibitem[Imai et~al.(2013)Imai, Ratkovic, et~al.]{imai2013estimating}
Imai, K., Ratkovic, M., et~al.
\newblock Estimating treatment effect heterogeneity in randomized program
  evaluation.
\newblock \emph{The Annals of Applied Statistics}, 7\penalty0 (1):\penalty0
  443--470, 2013.

\bibitem[Imbens \& Rubin(2015)Imbens and Rubin]{imbens2015causal}
Imbens, G.~W. and Rubin, D.~B.
\newblock \emph{Causal inference in statistics, social, and biomedical
  sciences}.
\newblock Cambridge University Press, 2015.

\bibitem[Jacob(2021)]{jacob2021cate}
Jacob, D.
\newblock Cate meets ml: Conditional average treatment effect and machine
  learning.
\newblock \emph{Digital Finance}, 3\penalty0 (2):\penalty0 99--148, 2021.

\bibitem[Juditsky \& Nemirovski(2000)Juditsky and
  Nemirovski]{juditsky2000functional}
Juditsky, A. and Nemirovski, A.
\newblock Functional aggregation for nonparametric regression.
\newblock \emph{The Annals of Statistics}, 28\penalty0 (3):\penalty0 681--712,
  2000.

\bibitem[Kallus(2018)]{kallus2018balanced}
Kallus, N.
\newblock Balanced policy evaluation and learning.
\newblock \emph{Advances in neural information processing systems}, 31, 2018.

\bibitem[Kennedy(2016)]{kennedy2016semiparametric}
Kennedy, E.~H.
\newblock Semiparametric theory and empirical processes in causal inference.
\newblock In \emph{Statistical causal inferences and their applications in
  public health research}, pp.\  141--167. Springer, 2016.

\bibitem[Kennedy(2017)]{kennedy2017semiparametric}
Kennedy, E.~H.
\newblock Semiparametric theory.
\newblock \emph{arXiv preprint arXiv:1709.06418}, 2017.

\bibitem[Kennedy(2020)]{kennedy2020optimal}
Kennedy, E.~H.
\newblock Optimal doubly robust estimation of heterogeneous causal effects.
\newblock \emph{arXiv preprint arXiv:2004.14497}, 2020.

\bibitem[Kennedy(2022)]{kennedy2022semiparametric}
Kennedy, E.~H.
\newblock Semiparametric doubly robust targeted double machine learning: a
  review.
\newblock \emph{arXiv preprint arXiv:2203.06469}, 2022.

\bibitem[Kennedy et~al.(2019)Kennedy, Harris, and Keele]{kennedy2019survivor}
Kennedy, E.~H., Harris, S., and Keele, L.~J.
\newblock Survivor-complier effects in the presence of selection on treatment,
  with application to a study of prompt icu admission.
\newblock \emph{Journal of the American Statistical Association}, 114\penalty0
  (525):\penalty0 93--104, 2019.

\bibitem[Kennedy et~al.(2020)Kennedy, Balakrishnan, and
  G’Sell]{kennedy2020sharp}
Kennedy, E.~H., Balakrishnan, S., and G’Sell, M.
\newblock Sharp instruments for classifying compliers and generalizing causal
  effects.
\newblock \emph{The Annals of Statistics}, 48\penalty0 (4):\penalty0
  2008--2030, 2020.

\bibitem[Kennedy et~al.(2021)Kennedy, Balakrishnan, and
  Wasserman]{kennedy2021semiparametric}
Kennedy, E.~H., Balakrishnan, S., and Wasserman, L.
\newblock Semiparametric counterfactual density estimation.
\newblock \emph{arXiv preprint arXiv:2102.12034}, 2021.

\bibitem[Kim et~al.(2018)Kim, Kim, and Kennedy]{kim2018causal}
Kim, K., Kim, J., and Kennedy, E.~H.
\newblock Causal effects based on distributional distances.
\newblock \emph{arXiv preprint arXiv:1806.02935}, 2018.

\bibitem[Kim et~al.(2022{\natexlab{a}})Kim, Kennedy, and
  Zubizarreta]{kim2022cfclassifier}
Kim, K., Kennedy, E., and Zubizarreta, J.~R.
\newblock Doubly robust counterfactual classification.
\newblock \emph{Advances in Neural Information Processing Systems}, 36,
  2022{\natexlab{a}}.

\bibitem[Kim et~al.(2022{\natexlab{b}})Kim, Mishler, and
  Zubizarreta]{kim2022counterfactualmv}
Kim, K., Mishler, A., and Zubizarreta, J.~R.
\newblock Counterfactual mean-variance optimization.
\newblock \emph{arXiv preprint arXiv:2209.09538}, 2022{\natexlab{b}}.

\bibitem[Kitagawa \& Tetenov(2018)Kitagawa and Tetenov]{kitagawa2018should}
Kitagawa, T. and Tetenov, A.
\newblock Who should be treated? empirical welfare maximization methods for
  treatment choice.
\newblock \emph{Econometrica}, 86\penalty0 (2):\penalty0 591--616, 2018.

\bibitem[Kleinberg et~al.(2016)Kleinberg, Mullainathan, and
  Raghavan]{kleinberg2016inherent}
Kleinberg, J., Mullainathan, S., and Raghavan, M.
\newblock Inherent trade-offs in the fair determination of risk scores.
\newblock \emph{arXiv preprint arXiv:1609.05807}, 2016.

\bibitem[K{\"u}nzel et~al.(2017)K{\"u}nzel, Sekhon, Bickel, and
  Yu]{kunzel2017meta}
K{\"u}nzel, S.~R., Sekhon, J.~S., Bickel, P.~J., and Yu, B.
\newblock Meta-learners for estimating heterogeneous treatment effects using
  machine learning.
\newblock \emph{arXiv preprint arXiv:1706.03461}, 2017.

\bibitem[Kusner et~al.(2017)Kusner, Loftus, Russell, and
  Silva]{kusner2017counterfactual}
Kusner, M.~J., Loftus, J.~R., Russell, C., and Silva, R.
\newblock Counterfactual fairness.
\newblock \emph{arXiv preprint arXiv:1703.06856}, 2017.

\bibitem[Levrard(2018)]{levrard2018quantization}
Levrard, C.
\newblock Quantization/clustering: when and why does $ k $-means work?
\newblock \emph{Journal de la soci{\'e}t{\'e} fran{\c{c}}aise de statistique},
  159\penalty0 (1):\penalty0 1--26, 2018.

\bibitem[Luedtke \& van~der Laan(2016{\natexlab{a}})Luedtke and van~der
  Laan]{luedtke2016optimal}
Luedtke, A.~R. and van~der Laan, M.~J.
\newblock Optimal individualized treatments in resource-limited settings.
\newblock \emph{The international journal of biostatistics}, 12\penalty0
  (1):\penalty0 283--303, 2016{\natexlab{a}}.

\bibitem[Luedtke \& van~der Laan(2016{\natexlab{b}})Luedtke and van~der
  Laan]{luedtke2016statistical}
Luedtke, A.~R. and van~der Laan, M.~J.
\newblock Statistical inference for the mean outcome under a possibly
  non-unique optimal treatment strategy.
\newblock \emph{Annals of statistics}, 44\penalty0 (2):\penalty0 713,
  2016{\natexlab{b}}.

\bibitem[Menon \& Williamson(2018)Menon and Williamson]{menon2018cost}
Menon, A.~K. and Williamson, R.~C.
\newblock The cost of fairness in binary classification.
\newblock In \emph{Conference on Fairness, Accountability and Transparency},
  pp.\  107--118. PMLR, 2018.

\bibitem[Mishler \& Kennedy(2021)Mishler and Kennedy]{mishler2021fade}
Mishler, A. and Kennedy, E.
\newblock Fade: Fair double ensemble learning for observable and counterfactual
  outcomes.
\newblock \emph{arXiv preprint arXiv:2109.00173}, 2021.

\bibitem[Mishler et~al.(2021)Mishler, Kennedy, and
  Chouldechova]{mishler2021fairness}
Mishler, A., Kennedy, E.~H., and Chouldechova, A.
\newblock Fairness in risk assessment instruments: Post-processing to achieve
  counterfactual equalized odds.
\newblock In \emph{Proceedings of the 2021 ACM Conference on Fairness,
  Accountability, and Transparency}, pp.\  386--400, 2021.

\bibitem[Murphy(2003)]{murphy2003optimal}
Murphy, S.~A.
\newblock Optimal dynamic treatment regimes.
\newblock \emph{Journal of the Royal Statistical Society: Series B (Statistical
  Methodology)}, 65\penalty0 (2):\penalty0 331--355, 2003.

\bibitem[Nabi \& Shpitser(2018)Nabi and Shpitser]{nabi2018fair}
Nabi, R. and Shpitser, I.
\newblock Fair inference on outcomes.
\newblock In \emph{Proceedings of the AAAI Conference on Artificial
  Intelligence}, volume~32, 2018.

\bibitem[Nabi et~al.(2019)Nabi, Malinsky, and Shpitser]{nabi2019learning}
Nabi, R., Malinsky, D., and Shpitser, I.
\newblock Learning optimal fair policies.
\newblock In \emph{International Conference on Machine Learning}, pp.\
  4674--4682. PMLR, 2019.

\bibitem[Nemirovski et~al.(2009)Nemirovski, Juditsky, Lan, and
  Shapiro]{nemirovski2009robust}
Nemirovski, A., Juditsky, A., Lan, G., and Shapiro, A.
\newblock Robust stochastic approximation approach to stochastic programming.
\newblock \emph{SIAM Journal on optimization}, 19\penalty0 (4):\penalty0
  1574--1609, 2009.

\bibitem[Neugebauer \& van~der Laan(2007)Neugebauer and van~der
  Laan]{neugebauer2007nonparametric}
Neugebauer, R. and van~der Laan, M.
\newblock Nonparametric causal effects based on marginal structural models.
\newblock \emph{Journal of Statistical Planning and Inference}, 137\penalty0
  (2):\penalty0 419--434, 2007.

\bibitem[Newey \& Robins(2018)Newey and Robins]{newey2018cross}
Newey, W.~K. and Robins, J.~R.
\newblock Cross-fitting and fast remainder rates for semiparametric estimation.
\newblock \emph{arXiv preprint arXiv:1801.09138}, 2018.

\bibitem[Nie \& Wager(2017)Nie and Wager]{nie2017quasi}
Nie, X. and Wager, S.
\newblock Quasi-oracle estimation of heterogeneous treatment effects.
\newblock \emph{arXiv preprint arXiv:1712.04912}, 2017.

\bibitem[Obermeyer et~al.(2019)Obermeyer, Powers, Vogeli, and
  Mullainathan]{obermeyer2019dissecting}
Obermeyer, Z., Powers, B., Vogeli, C., and Mullainathan, S.
\newblock Dissecting racial bias in an algorithm used to manage the health of
  populations.
\newblock \emph{Science}, 366\penalty0 (6464):\penalty0 447--453, 2019.

\bibitem[Orellana et~al.(2010)Orellana, Rotnitzky, and
  Robins]{orellana2010dynamic}
Orellana, L., Rotnitzky, A., and Robins, J.~M.
\newblock Dynamic regime marginal structural mean models for estimation of
  optimal dynamic treatment regimes, part i: main content.
\newblock \emph{The international journal of biostatistics}, 6\penalty0 (2),
  2010.

\bibitem[Robins et~al.(2008)Robins, Li, Tchetgen, and van~der
  Vaart]{robins2008higher}
Robins, J., Li, L., Tchetgen, E., and van~der Vaart, A.
\newblock Higher order influence functions and minimax estimation of nonlinear
  functionals.
\newblock In \emph{Probability and statistics: essays in honor of David A.
  Freedman}, pp.\  335--421. Institute of Mathematical Statistics, 2008.

\bibitem[Schnell et~al.(2016)Schnell, Tang, Offen, and
  Carlin]{schnell2016bayesian}
Schnell, P.~M., Tang, Q., Offen, W.~W., and Carlin, B.~P.
\newblock A bayesian credible subgroups approach to identifying patient
  subgroups with positive treatment effects.
\newblock \emph{Biometrics}, 72\penalty0 (4):\penalty0 1026--1036, 2016.

\bibitem[Schwartz et~al.(2022)Schwartz, Vassilev, Greene, Perine, Burt, Hall,
  et~al.]{schwartz2022towards}
Schwartz, R., Vassilev, A., Greene, K., Perine, L., Burt, A., Hall, P., et~al.
\newblock Towards a standard for identifying and managing bias in artificial
  intelligence.
\newblock 2022.

\bibitem[Semenova \& Chernozhukov(2021)Semenova and
  Chernozhukov]{semenova2021debiased}
Semenova, V. and Chernozhukov, V.
\newblock Debiased machine learning of conditional average treatment effects
  and other causal functions.
\newblock \emph{The Econometrics Journal}, 24\penalty0 (2):\penalty0 264--289,
  2021.

\bibitem[Shapiro(1993)]{shapiro1993asymptotic}
Shapiro, A.
\newblock Asymptotic behavior of optimal solutions in stochastic programming.
\newblock \emph{Mathematics of Operations Research}, 18\penalty0 (4):\penalty0
  829--845, 1993.

\bibitem[Shapiro et~al.(2014)Shapiro, Dentcheva, and
  Ruszczy{\'n}ski]{shapiro2014lectures}
Shapiro, A., Dentcheva, D., and Ruszczy{\'n}ski, A.
\newblock \emph{Lectures on stochastic programming: modeling and theory}.
\newblock SIAM, 2014.

\bibitem[Still(2018)]{still2018lectures}
Still, G.
\newblock Lectures on parametric optimization: An introduction.
\newblock \emph{Optimization Online}, 2018.

\bibitem[Tsybakov(2003)]{tsybakov2003optimal}
Tsybakov, A.~B.
\newblock Optimal rates of aggregation.
\newblock In \emph{Learning theory and kernel machines}, pp.\  303--313.
  Springer, 2003.

\bibitem[van~der Laan \& Luedtke(2014)van~der Laan and
  Luedtke]{van2014targeted}
van~der Laan, M.~J. and Luedtke, A.~R.
\newblock Targeted learning of an optimal dynamic treatment, and statistical
  inference for its mean outcome.
\newblock 2014.

\bibitem[Van~der Vaart(2000)]{van2000asymptotic}
Van~der Vaart, A.~W.
\newblock \emph{Asymptotic statistics}, volume~3.
\newblock Cambridge university press, 2000.

\bibitem[Viviano \& Bradic(2022)Viviano and Bradic]{viviano2022fair}
Viviano, D. and Bradic, J.
\newblock Fair policy targeting.
\newblock \emph{Journal of the American Statistical Association}, pp.\  1--37,
  2022.

\bibitem[Wager \& Athey(2018)Wager and Athey]{wager2018estimation}
Wager, S. and Athey, S.
\newblock Estimation and inference of heterogeneous treatment effects using
  random forests.
\newblock \emph{Journal of the American Statistical Association}, 113\penalty0
  (523):\penalty0 1228--1242, 2018.

\bibitem[Wang \& Rudin(2022)Wang and Rudin]{wang2022causal}
Wang, T. and Rudin, C.
\newblock Causal rule sets for identifying subgroups with enhanced treatment
  effects.
\newblock \emph{INFORMS Journal on Computing}, 34\penalty0 (3):\penalty0
  1626--1643, 2022.

\bibitem[White \& Fradella(2016)White and Fradella]{white2016stop}
White, M.~D. and Fradella, H.~F.
\newblock \emph{Stop and frisk: The use and abuse of a controversial policing
  tactic}.
\newblock NYU Press, 2016.

\bibitem[Zafar et~al.(2019)Zafar, Valera, Gomez-Rodriguez, and
  Gummadi]{zafar2019fairness}
Zafar, M.~B., Valera, I., Gomez-Rodriguez, M., and Gummadi, K.~P.
\newblock Fairness constraints: A flexible approach for fair classification.
\newblock \emph{The Journal of Machine Learning Research}, 20\penalty0
  (1):\penalty0 2737--2778, 2019.

\bibitem[Zhang et~al.(2012)Zhang, Tsiatis, Laber, and
  Davidian]{zhang2012robust}
Zhang, B., Tsiatis, A.~A., Laber, E.~B., and Davidian, M.
\newblock A robust method for estimating optimal treatment regimes.
\newblock \emph{Biometrics}, 68\penalty0 (4):\penalty0 1010--1018, 2012.

\bibitem[Zhang et~al.(2013)Zhang, Tsiatis, Laber, and
  Davidian]{zhang2013robust}
Zhang, B., Tsiatis, A.~A., Laber, E.~B., and Davidian, M.
\newblock Robust estimation of optimal dynamic treatment regimes for sequential
  treatment decisions.
\newblock \emph{Biometrika}, 100\penalty0 (3):\penalty0 681--694, 2013.

\bibitem[Zhao et~al.(2013)Zhao, Tian, Cai, Claggett, and
  Wei]{zhao2013effectively}
Zhao, L., Tian, L., Cai, T., Claggett, B., and Wei, L.-J.
\newblock Effectively selecting a target population for a future comparative
  study.
\newblock \emph{Journal of the American Statistical Association}, 108\penalty0
  (502):\penalty0 527--539, 2013.

\bibitem[Zhao(2019)]{zhao2019covariate}
Zhao, Q.
\newblock Covariate balancing propensity score by tailored loss functions.
\newblock \emph{The Annals of Statistics}, 47\penalty0 (2):\penalty0 965--993,
  2019.

\bibitem[Zheng \& Van Der~Laan(2010)Zheng and Van
  Der~Laan]{zheng2010asymptotic}
Zheng, W. and Van Der~Laan, M.~J.
\newblock Asymptotic theory for cross-valiyeard targeted maximum likelihood
  estimation.
\newblock \emph{Working Paper 273}, 2010.

\end{thebibliography}
\bibliographystyle{icml2023}

\newpage
\appendix
\onecolumn

\begin{center}
{\large\bf APPENDIX}
\end{center}
\vspace*{.1in}

\section{Fairness in Optimal Policy} \label{appsec:nonlinearity-example}
Even though a certain fairness criterion is met for a prediction $f(W)$ that we are concerned with, we might not be able to achieve the same level of fairness for $\mathbbm{1}\left\{f(W)>0\right\}$. To illustrate this, consider the criterion of independence and assume that for some $\sigma>0$,
\begin{align*}
    & f(W) \mid S=0 \sim N(-\delta, \sigma^2) \\
    & f(W) \mid S=1 \sim N(\delta, \sigma^2).
\end{align*}

The above setup naturally satisfies the independence fairness constraint:
\begin{align*}
    \left\vert \E\left\{f(W) \mid S=0 \right\} - \E\left\{f(W) \mid S=1 \right\} \right\Vert_{2,\Pb} \leq \delta.
\end{align*}

However we get
\begin{align*}
    & \left\vert \E\left[ \mathbbm{1}\left\{ f(W) > 0\right\} \mid S=1 \right] - \E\left[ \mathbbm{1}\left\{ \beta^\top\bm{b}(W) > 0 \right\} \mid S=0 \right] \right\Vert_{2,\Pb} \\
    & = \frac{1}{2}\left\{ \text{erf}\left( \frac{\delta}{\sigma} \right) - \text{erf}\left( -\frac{\delta}{\sigma} \right) \right\}.
\end{align*}

Hence, when $\sigma$ is large so there is sufficient overlap between the two conditional distributions, RHS of the above equation vanishes as $\delta \rightarrow 0$ and thereby the criterion of independence will likely to hold for $\mathbbm{1}\left\{ f(W) > 0\right\}$ as well. On the other extreme, when $\sigma$ is very small so there is no overlap, then RHS remains as $1/2$ and the same fairness criterion would not hold at all for $\mathbbm{1}\left\{ f(W) > 0\right\}$.

This shows that the following fairness condition alone
\[
    \left\vert \E\left\{\text{uf}(Z)\beta^\top\bm{b}(W)\right\} \right\vert \leq \delta
\] 
does not guarantee that the same condition holds for $\mathbbm{1}\left\{ \beta^\top\bm{b}(W) > 0 \right\}$ in general. 

However, this issue can be greatly alleviated by employing the fairness function such as \eqref{eqn:g-positive-balance} under the margin condition. In this case, the unfairness would differ at most by the order of $\delta + \left\Vert\tau(W) - \beta^\top\bm{b}(W)\right\Vert_{2,\Pb}$.

\section{Formal Definitions of the Regularity Conditions} \label{appsec:definitions-LICQ-SC}
Let $\mathcal{C}_{\text{fair}} \coloneqq \left\{\beta \bigm\vert \left\vert \E\left\{\text{uf}_j(Z)\beta^\top\bm{b}(W) \right\} \right\vert \leq \delta_j, \, j \in J \right\}$. For simplicity, define a set of inequality constraints $\mathcal{C}=\{\beta \mid g_j(\beta) \leq 0, 1 \leq j \leq M\}$ such that $\mathcal{C}=\{\mathcal{C}_{\text{fair}}, \mathcal{C}_{\text{lin}}\}$, where $\mathcal{C}_{\text{lin}}$ is a set of extra linear constraints
Then for a feasible point $\bar{\beta} \in \mathcal{C}$, we define the active index set.

\begin{definition}[Active set]
We define the active index set $J_0$ by
\[
J_0(\bar{\beta}) = \{1\leq j \leq M \mid g_j(\bar{\beta}) = 0 \}.
\]
\end{definition}

In what follows, we define LICQ and SC with respect to \eqref{eqn:true-opt-problem}.

\begin{definition}[LICQ]
\textit{Linear independence constraint qualification} (LICQ) is satisfied at $\bar{\beta} \in \mathcal{S}$ if the vectors $\nabla_\beta g_j(\bar{\beta})$, $j \in J_0(\bar{\beta})$ are linearly independent. 
\end{definition}

\begin{definition}[SC]
Let $L({\beta},{\gamma})$ be the Lagrangian.
\textit{Strict Complementarity} (SC) is satisfied at $\bar{\beta} \in \mathcal{S}$ if, with multipliers $\bar{\gamma}_j \geq 0$, $j \in J_0(\bar{\beta})$, the Karush-Kuhn-Tucker (KKT) condition 
\begin{align*}
    \nabla_\beta L(\bar{\beta},\bar{\gamma}) \coloneqq  \nabla_\beta \mathcal{L}(\bar{\beta}) + \underset{j \in J_0(\bar{\beta})}{\sum}\bar{\gamma}_j  \nabla_\beta g_j(\bar{\beta}) = 0,
\end{align*}
is satisfied such that
$
    \bar{\gamma}_j > 0, \forall j \in J_0(\bar{\beta}).
$    
\end{definition}
LICQ is arguably one of the most widely-used constraint qualifications that admit the first-order necessary conditions.
SC means that if the $j$-th inequality constraint is active then the corresponding dual variable is strictly positive, so exactly one of them is zero for each $1 \leq j \leq m$. SC is widely used in the optimization literature, particularly in the context of parametric optimization \citep[e.g.,][]{shapiro2014lectures}.

Next, we give the second-order condition that can replace Assumption \ref{assumption:A1}.

\begin{enumerate}[label={}, leftmargin=*]
        \item \customlabel{assumption:A1'}{(A1')} For each optimal solution $\beta^*$ in \eqref{eqn:true-opt-problem},
      \[	   
	        \varsigma^\top \E[\bm{b}(W)\bm{b}(W)^\top] \varsigma > 0, \quad \forall \varsigma \in \{\varsigma \in \R^k \mid \bm{b}(W)^\top\varsigma \leq 0, \text{uf}_j(Z)\bm{b}(W)^\top\varsigma \leq 0, j \in J_0(\beta^*) \} \setminus \{0\}.   
       \]
\end{enumerate}
Here, $J_0(\beta^*)$ is the active index set for $\beta^*$.
Assumption \ref{assumption:A1'} implies that $\beta^*$ is locally isolated, and is weaker than Assumption \ref{assumption:A1} which we used in the main text for simplicity. This guarantees that the quadratic growth condition holds at each $\beta^*$ \citep{shapiro2014lectures}.

\section{Estimation of Non-smooth Functionals under the Margin Condition} \label{app:margin-condition}

We first consider the fairness function which involves the non-smooth component of the form $\mathbbm{1}\{\vartheta(W) > 0\}$ for some $\vartheta$. We first provide the following useful lemma.

\begin{lemma}  \label{lem:non-smooth-mc}  
    For any functionals $\vartheta, \xi: \mathcal{Z} \rightarrow \R$, $\Vert \xi \Vert_{\infty,\Pb} < \infty$ and their estimates $\widehat{\xi}, \widehat{\vartheta}$, let $\psi = \Pb\left\{\phi_{\vartheta} \right\} \equiv \Pb\left\{\xi(Z) \mathbbm{1}\{\vartheta(W) > 0\} \right\}$ and consider the plug-in estimator $\widehat{\psi} = \Pn\left\{\widehat{\phi}_{{\vartheta}_{-B}} \right\} \equiv \frac{1}{K}\sum_{b=1}^K\Pn^b\left\{ \widehat{\phi}_{{\vartheta}_{-b}} \right\} = \frac{1}{K}\sum_{b=1}^K\Pn^b\left\{ \widehat{\xi}(Z) \mathbbm{1}\{\widehat{\vartheta}_{-b}(W) > 0\} \right\}$. Assume that
    \begin{enumerate}
        \item For any margin exponent $\alpha > 0$ and for all $t \geq 0$, the following margin condition holds:
        \begin{align*} 
            \Pb(\vert \vartheta(W)  \vert \leq t) \lesssim t^{\alpha} ,   
        \end{align*}        
        \item $\Pb\left[ \mathbbm{1}\{\widehat{\vartheta}(W) > 0\} \neq \mathbbm{1}\{\vartheta(W) > 0\}\right] + \Vert \widehat{\xi} - \xi \Vert_{2,\Pb} = o_\Pb(1).$
    \end{enumerate}
    Then, 
    \begin{align*}
        \widehat{\psi} - \psi = O_\Pb\left( \Vert \widehat{\xi} - \xi \Vert_{1,\Pb} + \Vert \widehat{\vartheta} - \vartheta \Vert_{\infty,\Pb}^\alpha + \frac{1}{\sqrt{n}} \right).
    \end{align*}
    If we further assume 
    \begin{enumerate}
    \item[3] 
        $
            \Vert \widehat{\xi} - \xi \Vert_{1,\Pb} + \Vert \widehat{\vartheta} - \vartheta \Vert_{\infty,\Pb}^\alpha = o_\Pb\left(n^{-\frac{1}{2}}\right),
        $
    \end{enumerate}
    then 
    \begin{align*}
        \sqrt{n}\left(\widehat{\psi} - \psi \right) \xrightarrow{d} N\left(0,\var(\phi)\right).
    \end{align*}    
\end{lemma}
\begin{proof}
We can write $\widehat{\psi} - \psi$ as
\begin{align*}
    \left(\Pn - \Pb\right)\phi_{\vartheta} + \left(\Pn - \Pb\right)\left(\widehat{\phi}_{{\vartheta}_{-B}} - \phi_{\vartheta} \right) +  \Pb\left(\widehat{\phi}_{{\vartheta}_{-B}} - \phi_{\vartheta} \right)
\end{align*}
For the second term, noting that  $\left(\Pn - \Pb\right)\left(\widehat{\phi}_{{\vartheta}_{-B}} - \phi_{\vartheta} \right) = \sum_{b=1}^K\left(\Pn - \Pb\right)\left\{(\widehat{\phi}_{{\vartheta}_{-b}} -  \phi_{\vartheta})(\mathbbm{1}(B=b)\right\}$ and $n \lesssim \frac{n}{K}$, we have
\begin{align*}
    \left\Vert \left(\Pn - \Pb\right)\left\{(\widehat{\phi}_{{\vartheta}_{-b}} -  \phi_{\vartheta})(\mathbbm{1}(B=b)\right\} \right\Vert_{2,\Pb} & \leq \left\Vert \widehat{\phi}_{{\vartheta}_{-b}} -  \phi_{\vartheta} \right\Vert_{2,\Pb} \lesssim \left\Vert \widehat{\phi}_{{\vartheta}} -  \phi_{\vartheta} \right\Vert_{2,\Pb} \\
    & = \left\Vert (\widehat{\xi} - {\xi}) \mathbbm{1}\{\widehat{\vartheta} > 0\} + \xi\left( \mathbbm{1}\{\widehat{\vartheta} > 0\} - \mathbbm{1}\{{\vartheta} > 0\} \right) \right\Vert_{2,\Pb} \\
    & \lesssim \Pb\left[ \mathbbm{1}\{\widehat{\vartheta}(W) > 0\} \neq \mathbbm{1}\{\vartheta(W) > 0\}\right] + \Vert \widehat{\xi} - \xi \Vert_{2,\Pb}.
\end{align*}
Hence, if $\Pb\left[ \mathbbm{1}\{\widehat{\vartheta}(W) > 0\} \neq \mathbbm{1}\{\vartheta(W) > 0\}\right] + \Vert \widehat{\xi} - \xi \Vert_{2,\Pb} = o_\Pb(1)$, then the the second term is $o_\Pb(\frac{1}{\sqrt{n}})$ by \citet[][Lemma 1]{kennedy2020sharp}.

The third term can be rewritten as
\begin{align*}
    \Pb\left\{(\widehat{\xi} - {\xi}) \mathbbm{1}(\widehat{\vartheta}_{-B} > 0)\right\} + \Pb\left[\xi\left\{\mathbbm{1}(\widehat{\vartheta}_{-B} > 0) - \mathbbm{1}({\vartheta} > 0) \right\} \right].
\end{align*}

The first term in the above display can be simply bounded as
\begin{align*}
    \Pb\left\{(\widehat{\xi} - {\xi}) \mathbbm{1}(\widehat{\vartheta}_{-B} > 0)\right\} \lesssim \Vert \widehat{\xi} - \xi \Vert_{1,\Pb}.
\end{align*}
For the second term, using the similar logic as before, we have that
\begin{align*}
    \Pb\left[\xi\left\{\mathbbm{1}(\widehat{\vartheta}_{-B} > 0) - \mathbbm{1}({\vartheta} > 0) \right\} \right] & \lesssim \Pb\left\vert \xi\left\{ \mathbbm{1}(\widehat{\vartheta} > 0) - \mathbbm{1}({\vartheta} > 0) \right\} \right\vert \\
    & \leq \Vert \xi \Vert_{\infty, \Pb} \Pb\left\vert \left\{ \mathbbm{1}(\widehat{\vartheta} > 0) - \mathbbm{1}({\vartheta} > 0) \right\} \right\vert \\
    & \lesssim \Pb\left( \vert \vartheta\vert \leq \vert \widehat{\vartheta} - \vartheta\vert \right) \\
    & \leq \Vert \widehat{\vartheta} - \vartheta \Vert_{\infty,\Pb}^\alpha,
\end{align*}
where the last inequality follows by the margin condition. Putting the two pieces together, we conclude
\begin{align*}
    \Pb\left(\widehat{\phi}_{{\vartheta}_{-B}} - \phi_{\vartheta} \right) & \lesssim \Vert \widehat{\xi} - \xi \Vert_{1,\Pb} + \Vert \widehat{\vartheta} - \vartheta \Vert_{\infty,\Pb}^\alpha     
\end{align*}
If this term is $o_\Pb\left(n^{-\frac{1}{2}}\right)$ (Condition 3), then the $\sqrt{n}$-consistency and asymptotic normality follows by the central limit theorem and Slutsky's theorem.
\end{proof}

Lemma \ref{lem:non-smooth-mc} gives corresponding convergence rates, as well as conditions under which $\widehat{\psi}$ is asymptotically normal and efficient. 

Now consider the unfairness estimand $\E\left\{\text{uf}_j(Z)\bm{b}(W)\right\}$ with the fairness function \eqref{eqn:g-positive-balance}
\begin{align*} 
    \text{uf}_j(Z)
    = \frac{(1-S)\mathbbm{1}\left\{\tau(W)>0\right\} }{\E\left[(1-S)\mathbbm{1}\left\{\tau(W)>0\right\}\right]}
    - \frac{S\mathbbm{1}\left\{\tau(W)>0\right\}}{\E\left[S \mathbbm{1}\left\{\tau(W)>0\right\}\right]}. 
\end{align*}

By Lemma \ref{lem:non-smooth-mc} and the continuous mapping theorem, one can deduce that the estimator 
\begin{align*}
    & \Pn\left\{\widehat{\text{uf}}_j\bm{b}(W) \right\} \\
    &= \frac{1}{K}\sum_{b=1}^K\Pn^b\Bigg(\Bigg[\frac{(1-S)\mathbbm{1}\left\{\widehat{\mu}_{1,-b}(W) - \widehat{\mu}_{0,-b}(W)>0\right\} }{\frac{1}{K}\sum_{b=1}^K\Pn^b\left[(1-S)\mathbbm{1}\left\{\widehat{\mu}_{1,-b}(W) - \widehat{\mu}_{0,-b}(W)>0\right\}\right]} \notag \\
    & \quad - \frac{S\mathbbm{1}\left\{\widehat{\mu}_{1,-b}(W) - \widehat{\mu}_{0,-b}(W)>0\right\}}{\frac{1}{K}\sum_{b=1}^K\Pn^b\left[S \mathbbm{1}\left\{\widehat{\mu}_{1,-b}(W) - \widehat{\mu}_{0,-b}(W)>0\right\}\right]}\Bigg]\bm{b}(W)\Bigg)
\end{align*}
attains the $\sqrt{n}$ rate of convergence if 
\begin{align*}
    & \max \Vert \widehat{\mu}_{a}- \mu_{a} \Vert_{\infty, \Pb}^{\alpha}=O_\Pb(n^{-\frac{1}{2}}), \\
    & \Pb\left[ \mathbbm{1}\{\widehat{\tau}(W) > 0\} \neq \mathbbm{1}\{\tau(W) > 0\}\right] = o_\Pb(1),
\end{align*}
and is asymptotically normal and efficient if it additionally holds that
\[
\max \Vert \widehat{\mu}_{a}- \mu_{a} \Vert_{\infty, \Pb}^{\alpha}=o_\Pb(n^{-\frac{1}{2}}).
\]

Importantly, Lemma \ref{lem:non-smooth-mc} could also be used for constructing efficient estimators for target functionals involving a non-smooth component of the more general form  $\mathbbm{1}\left\{\tau(W) \in \mathcal{I}\right\}$ for some interval (or union of intervals) $\mathcal{I}$ on $\R$, as long as $\mathbbm{1}\left\{\tau(W) \in \mathcal{I}\right\}$ can be expressed as a smooth function of $\left\{\mathbbm{1}\left\{\kappa_{1,j}\tau(W) + \kappa_{2,j} > 0\right\}\right\}_{j}$ for some constants $\kappa_{1,j}, \kappa_{2,j} \in \R$.

\section{Proofs}
\subsection{Proof of Theorem \ref{thm:asymptotics}}
One may rewrite \eqref{eqn:true-opt-problem} as the following nonlinear program with linear constraints
\begin{equation*}
\label{eqn:nlp-unobs}
\begin{aligned}
    & \underset{\beta \in \R^k}{\text{minimize}} \quad  f(\beta, T)  \\    
    & \text{subject to } \quad  C\beta \leq \delta,
\end{aligned}  
\end{equation*}
where 
\begin{align*}
    \delta &= \left[\delta_1, \ldots, \delta_m \right]^\top,\\
    T & = \{\E[\bm{b}(W)\bm{b}(W)^\top]_{i,j}, \E[(Y^1-Y^0)\bm{b}(W)]_i\}_{i,j=1}^k \in \R^{k(k+1)}, \\
    C &= [\E\{\text{uf}_1(Z)\bm{b}(W)\}, \ldots ,\E\{\text{uf}_m(Z)\bm{b}(W)\}]^\top \in \R^{m \times k}.
\end{align*}
According to our approximating program \eqref{eqn:approx-opt-problem} we have that,
\begin{align*}
    \widehat{T} &= \{\Pn[\bm{b}(W)\bm{b}(W)^\top]_{i,j}, \Pn\left\{  \left(\varphi_1(Z;\widehat{\vartheta}_{-B})-\varphi_0(Z;\widehat{\vartheta}_{-B})\right) \bm{b}(W) \right\}_i\}_{i,j=1}^k, \\
    \widehat{C} &= [\Pn\left\{\widehat{\text{uf}}_1\bm{b}(W)\right\}, \ldots, \Pn\left\{\widehat{\text{uf}}_m\bm{b}(W) \right\} ]^\top.
\end{align*}
Given the conditions \ref{assumption:A1} - \ref{assumption:A3}, \ref{assumption:A5}, by Lemma A.1 of \cite{kim2022cfclassifier}, it follows that
\begin{align*}
    & \Vert \widehat{T} - T \Vert_2 = O_\Pb\left(\max\Vert \widehat{\pi} - \pi \Vert_{2,\Pb} \Vert \widehat{\mu}_{a}- \mu_{a} \Vert_{2,\Pb} + n^{-\frac{1}{2}} \right), \\
    & \Vert \widehat{C} - C \Vert_F = O_\Pb\left(\max\Vert \widehat{\pi} - \pi \Vert_{2,\Pb} \Vert \widehat{\mu}_{a}- \mu_{a} \Vert_{2,\Pb} + n^{-\frac{1}{2}} \right).
\end{align*}
Further, under the conditions \ref{assumption:A1}, Theorem 3.1 of \cite{kim2022counterfactualmv} gives that
\begin{align*}
    \Vert \widehat{\beta} - \beta^* \Vert_2 &= O_\Pb\left(\max\Vert \widehat{\pi} - \pi \Vert_{2,\Pb} \Vert \widehat{\mu}_{a}- \mu_{a} \Vert_{2,\Pb} + n^{-\frac{1}{2}} \right).
\end{align*}

If we further assume the conditions \ref{assumption:A4},\ref{assumption:A6} so that every element in $\widehat{T}$, $\widehat{C}$ is indeed a $\sqrt{n}$-consistent and asymptotically normal estimator for the corresponding element in ${T}$, ${C}$, then the asymptotic normality with $\sqrt{n}$ rates can be obtained by Theorem 3.2 of \cite{kim2022counterfactualmv}, under the necessary regularity conditions described in Theorem \ref{thm:asymptotics}.

\subsection{Proof of Lemma \ref{lem:comparison-inequalities}} \label{appsec:proof-comparison-lemma}

\begin{proof}
The proof mimics the proofs of Lemma 5.1 and Lemma 5.2 in \cite{audibert2007fast}. To show the first inequality, note that
\begin{align*}
    \mathcal{U}(g^*) - \mathcal{U}(\widehat{g}) &= \Pb\left[ \Delta  \left( \mathbbm{1}\left\{\Delta > 0\right\} - \mathbbm{1}\left\{\widehat{\Delta} > 0\right\} \right)\right] \\
    & \leq \Pb\left[ \left\vert \Delta \right\vert  \left( \mathbbm{1}\left\{ \left\vert \Delta \right\vert \leq \left\vert \widehat{\Delta} - \Delta \right\vert \right\} \right)\right] \\
    & \leq \left\Vert \widehat{\Delta} - \Delta \right\Vert_{\infty, \Pb} \Pb\left\{ \left\vert \Delta \right\vert \leq \left\Vert \widehat{\Delta} - \Delta \right\Vert_{\infty, \Pb} \right\} \\
    & \lesssim \left\Vert \widehat{\Delta} - \Delta \right\Vert_{\infty, \Pb}^{\alpha + 1},
\end{align*}
where the first inequality follows by Lemma 1 of \cite{kennedy2020sharp} and the last by the margin condition.

Next, for any $t > 0$ we have
\begin{align*}
    \mathcal{U}(g^*) - \mathcal{U}(\widehat{g}) &\leq \Pb\left[ \left\vert \Delta \right\vert  \left( \mathbbm{1}\left\{ \left\vert \Delta \right\vert  \leq \left\vert \widehat{\Delta} - \Delta \right\vert \right\} \right) \mathbbm{1}\left\{ \left\vert \Delta \right\vert 
    \leq t \right\}\right] \\
    & \quad + \Pb\left[ \left\vert \Delta \right\vert  \left( \mathbbm{1}\left\{ \left\vert \Delta \right\vert  \leq \left\vert \widehat{\Delta} - \Delta \right\vert \right\} \right) \mathbbm{1}\left\{ \left\vert \Delta \right\vert 
    > t \right\}\right] \\
    & \leq \Pb\left[ \left\vert \widehat{\Delta} - \Delta \right\vert \mathbbm{1}\left\{ \left\vert \Delta \right\vert  \leq t \right\} \right] + \Pb\left[ \left\vert \widehat{\Delta} - \Delta \right\vert \mathbbm{1}\left\{ \left\vert \widehat{\Delta} - \Delta \right\vert  > t \right\} \right] \\
    & \leq \left\Vert \widehat{\Delta} - \Delta \right\Vert_{q,\Pb} \Pr\left\{ \left\vert \Delta \right\vert \leq t \right\}^\frac{q-1}{q} + \left\Vert \widehat{\Delta} - \Delta \right\Vert_{q,\Pb} \left(\frac{\Pb\vert \widehat{\Delta} - \Delta \vert^q}{t^q}\right)^{\frac{q-1}{q}} \\
    & \lesssim \left\Vert \widehat{\Delta} - \Delta \right\Vert_{q,\Pb} t^\frac{q-1}{q} + \frac{\left\Vert \widehat{\Delta} - \Delta \right\Vert_{q,\Pb}^q}{t^{q-1}},
\end{align*}
where the third inequality follows by the Hölder and Markov inequalities and the last by the margin condition. Now, the RHS in the last display is minimized when $t = O\left( \left\Vert \widehat{\Delta} - \Delta \right\Vert_{q,\Pb}^{\frac{q}{q+\alpha}} \right)$, yielding 
\begin{align*}
    \mathcal{U}(g^*) - \mathcal{U}(\widehat{g}) \lesssim  \left\Vert   \widehat{\Delta} - \Delta \right\Vert_{q,\Pb}^{\frac{q(1+\alpha)}{q+\alpha}}.
\end{align*}
\end{proof}

\subsection{Proof of Theorem \ref{thm:regret-bound}} \label{appsec:proof-tradeoff}

\begin{proof}
By the first inequality in Lemma \ref{lem:comparison-inequalities}, we have
\begin{align*}
    \mathcal{U}(g^*) - \mathcal{U}(\widehat{g}) & \lesssim \left\Vert \tau(W) - \widehat{\beta}^\top\bm{b}(W) \right\Vert_{\infty,\Pb}^{\alpha+1}
\end{align*}

Recall that $\beta^*$ and $\widetilde{\beta}$ are optimal solutions to \eqref{eqn:true-opt-problem} and \eqref{eqn:true-unconstrained}, respectively. 
Then $\forall a$, by the triangle and Cauchy–Schwarz inequalities,
\begin{align*}
    \left\vert \tau(W) - \widehat{\beta}^\top\bm{b}(W) \right\Vert_{2,\Pb}
    &\leq \left\vert \tau(W) - {\widetilde{\beta}}^\top\bm{b}(W)\right\Vert_{2,\Pb} + \left\Vert \bm{b}(W) \right\Vert_2 \left\{ \Vert \widetilde{\beta} - \beta^* \Vert_2 + \Vert \beta^* - \widehat{\beta} \Vert_2\right\}.
\end{align*}

Next, for some $\lambda \in \R^m$ consider the following Lagrange form of \eqref{eqn:true-opt-problem}.
\begin{equation}
\label{eqn:P-mu-lagrange}
\begin{aligned}
    & \underset{\beta \in \mathcal{B}}{\text{minimize}} \quad \E\left\{ \left( Y^1-Y^0- \beta^\top\bm{b}(W) \right)^2\right\} + \sum_{j=1}^m \lambda_j \left[ \beta^\top \E\left\{\text{uf}_j(Z)\bm{b}(W) \right\} \right]^2
\end{aligned}     
\end{equation}

Since our constraint set consists of only linear constraints strong duality holds, and there exists a dual solution $\lambda \geq 0$ to \eqref{eqn:true-opt-problem} such that any solution $\widetilde{\beta}$ in \eqref{eqn:true-opt-problem} minimizes
\begin{align*}
    \E\left\{ \left( Y^1-Y^0- \beta^\top\bm{b}(W) \right)^2\right\} + \sum_{j=1}^m \lambda_j  \left( \left[ \beta^\top \E\left\{\text{uf}_j(Z)\bm{b}(W) \right\} \right]^2 - \delta_j \right).
\end{align*}
Hence, $\widetilde{\beta}$ is a solution in \eqref{eqn:P-mu-lagrange} as well. 

Now, consider the following parametrized program 
\begin{equation}
\label{eqn:P-mu-param}
\begin{aligned}
    & \underset{\beta \in \mathcal{B}}{\text{minimize}} \quad f(\beta, \xi) \coloneqq \E\left\{ \left( Y^1-Y^0- \beta^\top\bm{b}(W) \right)^2\right\} + \sum_{j=1}^m \left(\beta^\top \xi_j\right)^2,
\end{aligned}     \tag{$\mathsf{P}(\xi)$}  
\end{equation}
with $\xi=(\xi_1,...,\xi_m)$, $\xi_j \in R^k$, and let $\beta(\xi)$ denote an optimal solution of $\mathsf{P}(\xi)$.
\eqref{eqn:true-unconstrained} and \eqref{eqn:P-mu-lagrange} correspond to \eqref{eqn:P-mu-param} with the parameter $\xi_j=0$ and $\xi_j= \sqrt{\lambda_j}\E\left\{\text{uf}_j(Z)\bm{b}(W) \right\}$, $\forall j$, respectively. 

By positive-definiteness of $\E[\bm{b}(W)\bm{b}(W)^\top]$, $\mathsf{P}(0)$ is strongly convex everywhere. Consequently the quadratic growth condition holds at $\beta(0)$ such that
\begin{align*}
    c \Vert \beta(0) - \beta' \Vert_2^2 \leq f(\beta', 0) - f\left(\beta(0), 0\right)
\end{align*}
for some constant $c>0$ and all $\beta' \in \R^k$. Hence,
\begin{align*}
    c \Vert \beta(0) - \beta(\xi) \Vert_2^2 & \leq f\left(\beta(\xi), 0\right) - f\left(\beta(0), 0\right) \\
    & = f\left(\beta(\xi), 0\right) - f\left(\beta(\xi), \xi \right) + f\left(\beta(\xi), \xi \right) - f\left(\beta(0), \xi \right) + f\left(\beta(0), \xi \right) - f\left(\beta(0), 0 \right) \\
    & \leq \left\{ \Vert  \beta(\xi) \Vert_2^2 + \Vert  \beta(0) \Vert_2^2 \right\} \Vert  \xi \Vert_2^2,
\end{align*}
where the last inequality follows by $f\left(\beta(\xi), \xi \right) - f\left(\beta(0), \xi \right) \leq 0$.

From the above argument, by the Cauchy–Schwarz and Jensen's inequality we have that
\begin{align*}
   \Vert \widetilde{\beta} - \beta^* \Vert_2 & \lesssim \left\Vert \E\left\{ \left(\sum_j \sqrt{\lambda_j} \text{uf}_j(Z) \right)\bm{b}(W)\right\} \right\Vert_2 \\
   & \leq \left\Vert \sum_j \sqrt{\lambda_j} \text{uf}_j(Z)\right\Vert_{2,\Pb} \sqrt{\sum_{i=1}^k \E\{{\bm{b}(W)_i}^2 \}} \\
   & \leq \left\Vert \sum_j \sqrt{\lambda_j} \text{uf}_j(Z)\right\Vert_{2,\Pb} \left\Vert \bm{b}(W) \right\Vert_{2,\Pb}.
\end{align*}

Further, under the conditions \ref{assumption:A1}-\ref{assumption:A3}, \ref{assumption:A5}, by Theorem \ref{thm:asymptotics}, it follows that
\begin{align*}
   \Vert \beta^* - \widehat{\beta} \Vert_2 = O_\Pb\left(n^{-\frac{1}{2}} +  \max_a\Vert \widehat{\pi}_a -\pi_a \Vert_{2,\Pb} \max \Vert \widehat{\tau} -\tau \Vert_{2,\Pb} \right).
\end{align*}

Since $0 < \alpha < \infty$, we obtain the desired result by putting the pieces together due to Minkowski's inequality.

The part (iii) follows by the exact same logic as used for the part (i).


The part (ii) immediately follows by the fact that
\begin{align*}
    \Pb\left\{\widehat{g}(W) \neq {g^*}(W) \right\} &= \Pb\left\{ \left\vert \widehat{g}(W) - {g^*}(W) \right\Vert_{2,\Pb}\right\}\\
    & \leq \Pb\left[ \mathbbm{1}\left\{ \left\vert \tau(W) \right\Vert_{2,\Pb} \leq \sum \left\vert \tau(W) - \widehat{\beta}^\top\bm{b}(W) \right\Vert_{2,\Pb} \right\} \right] \\ 
    & \lesssim \left\Vert \tau(W) - \widehat{\beta}^\top\bm{b}(W) \right\Vert_{\infty,\Pb}^{\alpha}.
\end{align*}


\end{proof}

\end{document}